\documentclass[a4paper,11pt]{article}
\usepackage{amssymb,amsmath,latexsym,amsthm}
\usepackage{color,mathrsfs}
\usepackage{url}
\usepackage{enumerate}
\usepackage[square,sort,comma,numbers]{natbib}
\usepackage{dsfont}
\usepackage{diagbox}
 \usepackage{geometry}
\usepackage{graphicx}
\usepackage{colortbl}

\usepackage{endnotes}

\renewcommand{\textbf}[1]{\begingroup\bfseries\mathversion{bold}#1\endgroup}

\usepackage{fancyhdr}
\usepackage{pstricks,pst-plot,pst-node,pstricks-add}

\newtheorem{thm}{Theorem}[section]
\newtheorem{defi}{Definition}[section]
\newtheorem{corollary}[thm]{Corollary}
\newtheorem{prop}[thm]{Proposition}

\newtheorem{lemma}[thm]{Lemma}

\theoremstyle{definition}
\newtheorem{remark}[thm]{Remark}

\newtheorem{examples}[thm]{Examples}

\newcommand{\argmin}{\mathop{\rm argmin}\nolimits}

\newcommand{\R}{\mathbb R}

\newcommand{\Z}{\mathbb Z}
\newcommand{\N}{\mathbb N}

\numberwithin{equation}{section}

\def\XXint#1#2#3{{\setbox0=\hbox{$#1{#2#3}{\int}$}
    \vcenter{\hbox{$#2#3$}}\kern-.5\wd0}}

\allowdisplaybreaks
\date{date}
\begin{document}
\title{Local variational study of 2d lattice energies and application to Lennard-Jones type interactions}
\author{Laurent B\'{e}termin\\ \\
QMATH, Department of Mathematical Sciences, University of Copenhagen,\\ Universitetsparken 5, DK-2100 Copenhagen \O, Denmark.\\ \texttt{betermin@math.ku.dk}. ORCID id: 0000-0003-4070-3344 }
\date\today
\maketitle

\begin{abstract}
In this paper, we focus on finite Bravais lattice energies per point in two dimensions. We compute the first and second derivatives of these energies. We prove that the Hessian at the square and the triangular lattice are diagonal and we give simple sufficient conditions for the local minimality of these lattices. Furthermore, we apply our result to Lennard-Jones type interacting potentials that appear to be accurate in many physical and biological models. The goal of this investigation is to understand how the minimum of the Lennard-Jones lattice energy varies with respect to the density of the points. Considering the lattices of fixed area $A$, we find the maximal open set to which $A$ must belong so that the triangular lattice is a minimizer (resp. a maximizer) among lattices of area $A$. Similarly, we find the maximal open set to which $A$ must belong so that the square lattice is a minimizer (resp. a saddle point). Finally, we present a complete conjecture, based on numerical investigations and rigorous results among rhombic and rectangular lattices, for the minimality of the classical Lennard-Jones energy per point with respect to its area. In particular, we prove that the minimizer is a rectangular lattice if the area is sufficiently large.
\end{abstract}



\section{Introduction and main results}

Many two-dimensional interacting systems exhibit a periodic order, especially at very low temperature. Ginzburg-Landau vortices, Wigner crystal, Bose-Einstein Condensates and graphene sheets are well-known examples of the manifestation of the so-called ``crystallization": the ground state of their interaction energy tends to be a lattice structure. As explained in \cite{Blanc:2015yu}, the mathematical justification of such phenomena is a challenging problem, even in two dimensions. If we consider only points interacting in the whole plane via a radially symmetric potential, only few results are known \cite{Rad2,Rad3,Crystal,Luca:2016aa}. In each case, the interaction potential can be viewed as an approximation of a Lennard-Jones type potential defined for $r>0$ by
\begin{equation}\label{DEFLJINTRO}
V_{a,t}^{LJ}(r):=\frac{a_2}{r^{t_2}}-\frac{a_1}{r^{t_1}},\quad a=(a_1,a_2)\in (0,+\infty)^2,\quad t=(t_1,t_2)\quad \textnormal{s.t.}\quad 1<t_1<t_2.
\end{equation}
For these models, the global optimality of a triangular lattice, i.e. a dilation of $\Z(1,0)\oplus\Z(1/2,\sqrt{3}/2)$ (see \eqref{deftriA}), is proved as the number of points goes to infinity, in the sense of the thermodynamic limit (i.e. for the average energy).

\medskip

The Lennard-Jones type potentials, also called ``Mie potentials" appear to be a good model for the interactions in a solid (see e.g. \cite[p. 624]{CondensMatter}), or to describe hydrogen bonds \cite{GelinKarplus}. They are also used in mathematical biology to study social aggregation \cite{MEKBS}. Furthermore, they can be seen as the difference of two homogeneous functions, which allows to efficiently use the change of scale. It was also pointed out by Ventevogel and Nijboer \cite{VN2} that the crystallization (at least in dimension $1$) is easier to prove for Lennard-Jones type potentials than for other long-range decreasing-increasing functions like the Morse potential $V_M(r)=Ae^{-\lambda r}-Be^{-\mu r}$, $\lambda>\mu>0$.

\medskip

In this paper, we choose the following approach. Let $f\in C^2((0,+\infty))$ be a radially symmetric potential such that, for any $k\in\{0,1,2\}$, $f^{(k)}$ is integrable at infinity (see \ref{deff}). Assuming that the ground state of the associated energy per point $E_f$ is a Bravais lattice $L=\Z u\oplus \Z v\subset \R^2$, what is the minimizer of 
\begin{equation*}
E_f[L]:=\sum_{p\in L\backslash \{0\}} f(|p|^2)
\end{equation*}
among those lattices? Thus, we restrict the minimization problem to the simplest periodic sets of points. This problem, also discussed in \cite[Sect. 2.5]{Blanc:2015yu}, appears to be an interesting first step in order to obtain information about the ground state of the energy. For example, it is possible to exclude a large majority of lattice structures from the list of possible minimizers thanks to it. Furthermore, any optimality result for $L\mapsto E_f[L]$ supports the associated crystallization conjecture for particles interacting through $f$.

\medskip

This problem of minimizing energies among Bravais lattices has been investigated for inverse power laws $f(r)=r^{-s/2}$, $s>0$ \cite{Rankin,Eno2,Cassels,Diananda} and gaussian potentials $f(r)=e^{-\pi \alpha r}$ \cite{Mont} where the corresponding energies are respectively the Epstein zeta function and the lattice theta function defined by
\begin{equation}\label{defEpsttheta}
\zeta_L(s)=\sum_{p\in L\backslash \{0\}} \frac{1}{|x|^{s}}, \quad \textnormal{and}\quad \theta_L(\alpha):=\sum_{p\in L} e^{-\pi \alpha |p|^2}.
\end{equation}
In each case, the triangular lattice is the unique minimizer at any scale, i.e. among all Bravais lattices of any fixed density. These results were used in Mathematical Physics to prove, for instance, the optimality of the triangular lattice for Bose-Einstein condensates \cite{AftBN} and Ginzburg-Landau vortices \cite{Sandier_Serfaty} among Bravais lattices of fixed density, supporting some important conjectures (see e.g. \cite{Betermin:2014rr} for a connection between Sandier-Serfaty's Vortices Conjecture and Smale's $7^{th}$ problem).

\medskip

This kind of problem contains a high degree of nonlinearity. Indeed, as explained in \cite{Rankin,Mont}, a two-dimensional Bravais lattice $L$ can be parametrized by three real numbers $(x,y,A)$ where 
\begin{equation*}
(x,y)\in \mathcal{D}:=\left\{ (x,y)\in \R^2 ; 0\leq x\leq 1/2, y>0, x^2+y^2\geq 1 \right\},\quad \textnormal{and}\quad A>0.
\end{equation*}
The set $\mathcal{D}$ is called the half modular domain, $(x,y)$ parametrizes the ``shape" of the lattice, and $A$ is its area (the area of its unit cell $\R^2/L$, $A=|u \wedge v |$, also called the covolume of $L$). Four types of lattices play an important role:
\begin{itemize}
\item The triangular lattice of area $A$ (see Figure \ref{TriLattice}), $\Lambda_A:=\sqrt{\frac{2A}{\sqrt{A}}}\left[\Z(1,0)\oplus \Z\left( \frac{1}{2},\frac{\sqrt{3}}{2} \right)  \right]$ parametrized by $(1/2,\sqrt{3}/2,A)$.
\begin{figure}[!h]
\centering
\includegraphics[width=5cm,trim = 1mm 0cm 1mm 0cm, clip]{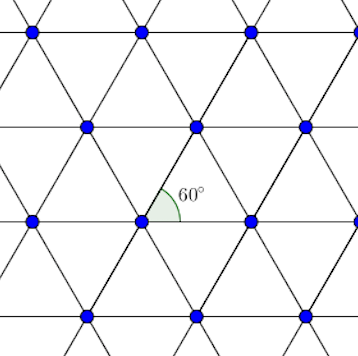} 
\caption{Triangular lattice.}
\label{TriLattice}
\end{figure}
\item The square lattice of area $A$ (see Figure \ref{SqLattice}), $\sqrt{A}\Z^2$ parametrized by $(0,1,A)$.
\begin{figure}[!h]
\centering
\includegraphics[width=5cm]{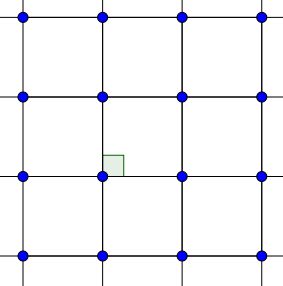} 
\caption{Square lattice.}
\label{SqLattice}
\end{figure}
\item The rhombic lattices (see Figure \ref{RhLattice}) parametrized by $(\cos \theta, \sin \theta, A)$, $60^\circ\leq \theta\leq 90^\circ$, having their generating vectors $u,v$ such that $|u|=|v|$ and $(\widehat{u,v})=\theta$ (including the triangular lattice if $\theta=60^\circ$ and the square lattice if $\theta=90^\circ$).
\begin{figure}[!h]
\centering
\includegraphics[width=5cm]{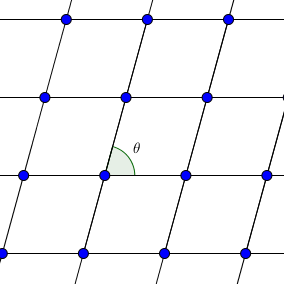} 
\caption{Rhombic lattice.}
\label{RhLattice}
\end{figure}
\item The rectangular lattices $\sqrt{A}\left[ \Z\left(\frac{1}{\sqrt{y}},0  \right)\oplus \Z \left(0,\sqrt{y}  \right) \right]$ (see Figure \ref{RectLattice}) parametrized by $(0,y,A)$, $y\geq 1$ (including the square lattice if $y=1$).
\begin{figure}[!h]
\centering
\includegraphics[width=5cm,trim = 0mm 0cm 1mm 0cm, clip]{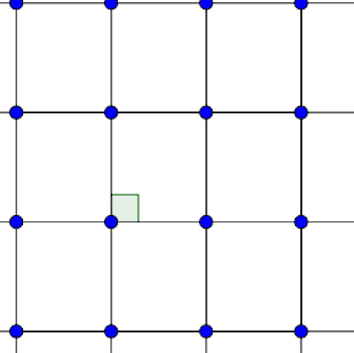} 
\caption{Rectangular lattice.}
\label{RectLattice}
\end{figure}
\end{itemize} 
Thus, the (nonlinear) energy of the lattice can be written as
\begin{equation*}
E_f[L]=E_f(x,y,A)=\sum_{(m,n)\in \Z^2\backslash (0,0)} f\left(A\left[\frac{1}{y}(m+xn)^2+yn^2  \right]   \right).
\end{equation*}
Then, once $A>0$ is fixed, the equations of the critical points of $(x,y)\mapsto E_f(x,y,A)$ are nonlinear (see Proposition \ref{deriv}), as is the behaviour of the minimizer with respect to $A$, i.e. $A\mapsto \argmin_{(x,y)\in \mathcal{D}}E_f(x,y,A)$. In this paper, we study the local minima of $(x,y)\mapsto E_f(x,y,A)$ for fixed $A>0$ in order to get a better understanding of the multistability of this system. More precisely, we want to characterize the values of $A$ such that the triangular lattice (resp. square lattice) given by $(x,y)=(1/2,\sqrt{3}/2)$ (resp. $(x,y)=(0,1)$) is a local minimizer, maximizer or saddle point of the energy. 

\medskip

It appears that finding the global minimum of $L\mapsto E_f[L]$ among all the Bravais lattices (with or without a fixed density) is a difficult problem, especially in larger dimension than two (see e.g. \cite{SarStromb}). In \cite{BeterminPetrache}, the author and Petrache developed some methods to study the lattice theta function defined by \eqref{defEpsttheta} based on dimension reduction. Using local minimality results in dimension $2$ for the triangular lattice and the square lattice, the local minimality of the Face-Centred-Cubic (FCC) and the Body-Centred-Cubic (BCC) lattices have been proved, for $L\mapsto \theta_L(\alpha)$, with respect to the parameter $\alpha$. Thus, the local stability of the triangular (or square) layers of a three-dimensional lattice can give information about the stability of the whole lattice. In \cite{Beterminlocal3d}, we have studied the same problem of local stability in dimension $3$ for cubic lattices. Some analogies with this work will be pointed out in this paper (see Table \ref{table-2d3d}). Furthermore, general results in any dimension have been proved by Coulangeon and Sch\"urmann  \cite{Coulangeon:kx,Coulangeon:2010uq,CoulSchurm2018} connecting local minimality and metric properties of some lattices.

\medskip

In \cite{Betermin:2014fy,BetTheta15}, the author and Zhang studied $E_f$ where $f=V_{a,t}^{LJ}$ is a Lennard-Jones type potential defined by \eqref{DEFLJINTRO}. Using theta functions, the triangular lattice has been proved to be a global minimum for small parameters $(t_1,t_2)$ and a minimizer at high density (i.e. for small $A$) for any parameters. It was also shown that the triangular lattice is no longer a minimizer of $E_{V_{a,t}^{LJ}}$ if $A$ is large enough. These results are summarized in the following theorem, where $\Gamma$ is the usual Gamma function:
\begin{thm}[\cite{Betermin:2014fy,BetTheta15}] Let $V_{a,t}^{LJ}$ be defined by \eqref{DEFLJINTRO}, then:
\begin{enumerate}
\item for any $A\leq \pi\left(\frac{a_2\Gamma(t_1)}{a_1\Gamma(t_2)}  \right)^{\frac{1}{t_2-t_1}}$, the triangular lattice $\Lambda_A$ is the unique minimizer, up to rotation, of $L\mapsto E_{V_{a,t}^{LJ}}[L]$ among Bravais lattices of fixed area $A$;
\item the triangular lattice $\Lambda_A$ is a minimizer of $E_{V_{a,t}^{LJ}}$ among Bravais lattices of fixed area $A$ if and only if
\begin{equation}\label{UBLJ}
A\leq \inf_{|L|=1,L\neq \Lambda_1}\left(\frac{a_2\left( \zeta_L(2t_2)-\zeta_{\Lambda_1}(2t_2) \right)}{a_1\left( \zeta_L(2t_1)-\zeta_{\Lambda_1}(2t_1) \right)}   \right)^{\frac{1}{t_2-t_1}},
\end{equation}
where the infimum is taken over the Bravais lattices $L$ of area $1$ and the Epstein zeta function $\zeta_L$ is defined by \eqref{defEpsttheta};
\item if $\pi^{-t_2}\Gamma(t_2)t_2\leq \pi^{-t_1}\Gamma(t_1)t_1$, then the minimizer of $L\mapsto E_{V_{a,t}^{LJ}}[L]$ among all the Bravais lattices (without a density constraint) is unique and triangular.
\end{enumerate}
\end{thm}

We remark that point 2. implies the non-minimality of $\Lambda_A$ if $A$ is sufficiently large. Hence, in \cite{Betermin:2014fy}, we numerically computed that the right side term of \eqref{UBLJ} in the classical case $V(r)=r^{-6}-2r^{-3}$ (i.e. $a=(2,1)$, $t=(3,6)$) is 
$$
A_{BZ}:= \inf_{|L|=1\atop L\neq \Lambda_1} \left(\frac{\zeta_L(12)-\zeta_{\Lambda_1}(12)}{2(\zeta_L(6)-\zeta_{\Lambda_1}(6))}  \right)^{1/3}\approx 1.138.
$$ 
Furthermore, we conjectured that the square lattice must be a minimizer for some values of the area (in an interval) larger than $A_{BZ}$. Obviously, our method, based on the global optimality of the triangular lattice for the theta function $L\mapsto \theta_L(\alpha)$ defined by \eqref{defEpsttheta}, was not adapted to prove the optimality of another lattice (square, rectangular or rhombic). Thus, the goal of this paper is to study $L\mapsto E_{V_{a,t}^{LJ}}[L]$ locally in order to get more information about the optimality of the triangular and the square lattice, and then to specify our conjecture -- in the classical case -- about the minimizers of $(x,y)\mapsto E_f(x,y,A)$ with respect to the area $A$.

\medskip

The fact that the square and the triangular lattices are critical points of $(x,y)\mapsto E_f(x,y,A)$ for any fixed $A$ is a consequence of the lattice symmetries, as proved in \cite[Thm. 4.4.(1)]{Coulangeon:2010uq}. We give an alternative proof of this result in Propositions \ref{D1xy} and \ref{triangCM}. Then, the Hessians of $(x,y)\mapsto E_f(x,y,A)$ appear to be diagonal for both lattices (see Corollary \ref{Hessienne01} and Proposition \ref{derivtriangular}), which is again a consequence of the symmetries. Moreover, the Hessian for the triangular lattice is a multiple of the identity. Thus, it is clear that, for any classical interacting potential $f$ (constructed with exponentials, inverse power laws or other classical functions), the triangular lattice is, for almost every $A>0$, a local minimizer or a local maximizer (see Corollary \ref{coranalytic}). Furthermore, we get the following result:

\begin{thm}[See Thm. \ref{THmain1} and Thm. \ref{THmain2} below] \label{THintro3}
We define the following sums:
\begin{center}
\begin{minipage}[l]{0.4\linewidth}
\begin{align*}
&S_1(s)=\sum_{m,n}\frac{m^4}{(m^2+mn+n^2)^s},\\
&S_3(s)=\sum_{m,n}\frac{m^2n^2}{(m^2+n^2)^s},
\end{align*}
\end{minipage}
\begin{minipage}[c]{0.5\linewidth}
\begin{align*}
&S_2(s)=\sum_{m,n}\frac{m^2}{(m^2+n^2)^s},\\
&S_4(s)=\sum_{m,n}\frac{(n^2-m^2)^2}{(m^2+n^2)^s}.
\end{align*}
\end{minipage}
\end{center}
\textnormal{\textbf{Part A: Local optimality of the triangular lattice.}} For any $(a,t)$ as in \eqref{DEFLJINTRO}, let 
$$
A_{0}:=\frac{\sqrt{3}}{2}\left( \frac{a_2 t_2(t_2-1)S_1(t_2+2)}{a_1t_1(t_1-1)S_1(t_1+2)} \right)^{\frac{1}{t_2-t_1}},
$$
then we have:
\begin{enumerate}
\item if $A<A_0$, then $\left(\frac{1}{2},\frac{\sqrt{3}}{2}  \right)$ is a local minimizer of $(x,y)\mapsto E_{V_{a,t}^{LJ}}(x,y,A)$;
\item if $A>A_0$, then $\left(\frac{1}{2},\frac{\sqrt{3}}{2}  \right)$ is a local maximizer of $(x,y)\mapsto E_{V_{a,t}^{LJ}}(x,y,A)$.
\end{enumerate} 
\textnormal{\textbf{Part B. Local optimality of the square lattice.}} Let
$$
g(s)=S_2(s+1)-2(s+1)S_3(s+2),\quad k(s)=(s+1)S_4(s+2)-2S_2(s+1),
$$
and define
$$A_1:=\left(\frac{a_2t_2g(t_2)}{a_1t_1g(t_1)}  \right)^{\frac{1}{t_2-t_1}} \quad \textnormal{and}\quad A_2:=\left(\frac{a_2t_2k(t_2)}{a_1t_1k(t_1)}  \right)^{\frac{1}{t_2-t_1}}.
$$
It holds:
\begin{enumerate}
\item if $A_1<A<A_2$, then $(0,1)$ is a local minimizer of $(x,y) \mapsto E_{V_{a,t}^{LJ}}(x,y,A)$;
\item if $A\not\in [A_1,A_2]$, then $(0,1)$ is a saddle point of $(x,y) \mapsto E_{V_{a,t}^{LJ}}(x,y,A)$.
\end{enumerate}
\end{thm}

\begin{remark}
It seems difficult to compare $A_0$ and $A_1$ in general. However, in the classical case $a=(2,1)$, $t=(3,6)$, we have $A_1\approx 1.143< A_0\approx 1.152$. Therefore, for any $A\in (A_1,A_0)$, the triangular and the square lattices are local minimizers of the energy.
\end{remark}

For the classical Lennard-Jones interaction $V$, i.e. $a=(2,1)$ and $t=(3,6)$, we can compare this result with \cite{Beterminlocal3d} where the local optimality of the cubic lattice $\Z^3$, the BCC lattice and the FCC lattice has been studied. In Table \ref{table-2d3d}, we have summarized our results in terms of scaling parameter $\ell=V^{1/3}$ in the three-dimensional case ($V$ being the volume of the unit cell) and $\ell=\sqrt{A}$ in the two-dimensional case. Comparing the optimality of $\Z^2$ and $\Z^3$, it turns out that the results are similar -- there is an interval of area/volume where they are local minimizers and saddle points outside -- and the numerical values for the bounds of these intervals are very close to each other (once the volumes and areas are converted to lengths). The main difference is for the FCC and BCC lattices, which can be viewed as layering of triangular lattices (see e.g. \cite[Sect. 2.B.]{BeterminPetrache}). There is a small volume region where the FCC and BCC lattices are saddle points, which is not the case for the triangular lattice. This phenomenon, as well as the fact that all the values are different, is obviously explained by the additional dimension. However, the numerical values of the volume for the local minimality and maximality of these lattices are also very close to each other. Thus, as explained in \cite{BeterminPetrache} for the lattice theta function, it seems that the study of the local minimality/maximality of one layer of a three-dimensional lattice gives a rather accurate information on the local minimality/maximality of the whole lattice for the classical Lennard-Jones energy.

\begin{table}[!h]
\centering
\begin{tabular}{|c|c|c|c|c|c|}
\hline
 & \textbf{$\ell\Z^2$} & \textbf{$\Lambda_{\ell^2}$} & \textbf{$\ell\Z^3$} & \textbf{$\ell$FCC  and $\ell$BCC} \\
 \hline
 \textbf{Local minimum} & $1.069<\ell < 1.126$ & $0<\ell < 1.073$ & $1.063<\ell<1.104$ & $0< \ell <1.029$\\
\hline
\textbf{Local maximum} & No & $\ell > 1.073$ &  No & $\ell>1.095$\\
\hline
\textbf{Saddle point} & $\ell \not\in (1.069,1.126)$ & No & $\ell\not\in (1.063,1.104)$ & $1.029<\ell<1.095$\\
\hline
\end{tabular}
\caption{Local optimality results for $2d$ and $3d$ lattices in terms of the scaling parameter $\ell\in \{\sqrt{A},V^{1/3}\}$.}
\label{table-2d3d}
\end{table}

Furthermore, again in the classical case, we give a complete conjecture, improving that of \cite{Betermin:2014fy}, based on our previous result and numerical simulations. A summary of this conjecture is given in Figure \ref{CONJINTRO} and justified in Sections \ref{secrhomb} and \ref{rectangular}. Furthermore, we summarize in Table \ref{tableintro} what is precisely conjectured, proved and numerically checked.

\begin{figure}[!h]
\centering
\includegraphics[width=16cm]{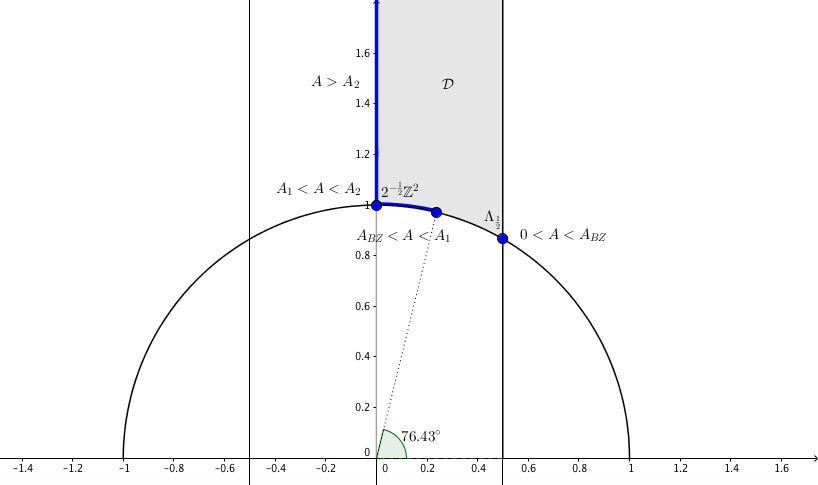} 
\caption{\textbf{Conjecture} about the minimization of $(x,y)\mapsto E_V(x,y,A)$ with respect to $A$. (1) If $0<A<A_{BZ}\approx 1.138$, then the minimizer is triangular. (2) If $A_{BZ}<A<A_1\approx 1.143$, then the minimizer is a rhombic lattice with an angle covering monotonically and continuously the interval $[76.43^\circ,90^\circ)$. (3) If $A_1<A<A_2\approx 1.268$, then the minimizer is a square lattice. (4) If $A>A_2$, then the minimizer is a rectangular lattice which degenerates (the primitive cell is more and more thin, see Figure \ref{Conjrectdeg}) as $A\to +\infty$.}
\label{CONJINTRO}
\end{figure}

\begin{table}[!h]
\centering
\begin{tabular}{|c|c|c|}
\hline
\textbf{Area $A$} & \textbf{Min of $L_A\mapsto E_V[L_A]$} & \textbf{Status}\\
\hline
$0<A<\frac{\pi}{(120)^{1/3}}\approx 0.637$ & triangular & proved in \cite{Betermin:2014fy}\\
\hline
$\frac{\pi}{(120)^{1/3}}<A<A_{BZ}\approx 1.138$ & triangular & num. + loc. min. proved in Th. \ref{THmain1}\\
\hline
$A_{BZ}<A<A_1\approx 1.143$ & rhombic & num. \\
\hline
$A_1<A<A_2\approx 1.268$ & square & num. + loc. min. proved in Th. \ref{THmain2} \\
\hline
$A>A_2$ & rectangular & num., proved for large $A$ in Prop \ref{Rankinmethod} \\
\hline
\end{tabular}
\caption{Summary of our works. The abbreviations ``num." and ``loc. min" mean ``numerically checked" and ``local minimality".}
\label{tableintro}
\end{table}

Using a method of Rankin \cite{Rankin} and bounding the minimizer of $y\mapsto E_V(0,y,A)$ in terms of $A$, we show the following result, which partially proves the point $(4)$ of our Conjecture in Figure \ref{CONJINTRO}:

\begin{thm}[See Prop. \ref{degrect} and Prop. \ref{Rankinmethod} below]\label{THintro4}
Let $V(r)=\frac{1}{r^6}-\frac{2}{r^3}$, then there exists $\tilde{A}>0$ such that for any $A>\tilde{A}$, the minimizer $(x_A,y_A)$ of $(x,y)\mapsto E_{V}(x,y,A)$ is such that $x_A=0$, $y_A\geq 1$. Furthermore, we have $\displaystyle \lim_{A\to +\infty} y_A=+\infty$.
\end{thm}

\begin{figure}[!h]
\centering
\includegraphics[width=5cm]{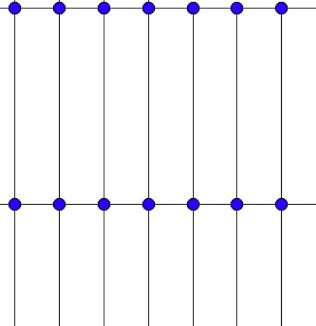} 
\caption{Degeneracy of the rectangular minimizer.}
\label{Conjrectdeg}
\end{figure}

The conjecture we give is actually comparable to the numerical study of Ho and Mueller \cite[Fig. 1 and 2]{Mueller:2002aa} about the two-component Bose-Einstein Condensates (see also the review \cite[Fig. 16]{ReviewvorticesBEC}). Indeed, $A^6 E_V(x,y,A)=\zeta_L(12)-A^3\zeta_L(6)$ is the sum of two terms with opposite behavior. The first one, $\zeta_L(12)$, is minimized by the triangular lattice and the second one, $-A^3\zeta_L(6)$, admits a degenerate minimizer. Thus, if $A$ is small enough, then the global behavior of $E_V(x,y,A)$ is similar to the one of $\zeta_L(12)$, i.e. the triangular lattice is the minimizer, and if $A$ is large enough, the minimizer must degenerate. That is precisely what appears in our results and numerics. We find exactly the same kind of terms in the energy studied by Ho and Mueller (see Section \ref{summary} for more explanations). Furthermore, according to the one-dimensional study of Ventevogel and Nijboer \cite{VN2}, it is reasonable to think that the behavior of $A\mapsto \argmin_{L_A} E_f[L_A]$ should be the same for a large class of potentials with a well and that are equivalent to a completely monotone function (see Definition \ref{defCMfct}) in a neighbourhood of the origin (which ensures the optimality of the triangular lattice at high density). By a Mellin transform argument (see e.g. \cite[Eq. (27)]{Coulangeon:kx}), it may be possible to derive some properties from the Lennard-Jones type potentials to some more general physically-relevant radially symmetric potentials, and maybe also for Ho-Mueller's lattice energy. In other words, we believe that the properties of the minimizers depicted in Figure \ref{CONJINTRO} could be universal for a large class of potentials and energies involving repulsion at short distance, equilibrium at finite distance and decay of the interaction at large distance to $0$. For example, another recent work \cite{Samaj12,TrizacWigner16} on Wigner bilayers presents a surprising similarity.

\medskip

Another question would be to generalize these results to non-radially symmetric potentials. For instance, it is possible to add an angle-dependent term to the energy. These types of potentials were used in \cite{ELi,Stef2} to prove the minimality of a honeycomb structure and in \cite{Stef1} to show the optimality of a square lattice configuration. In both cases, the radial part of the energy is an approximation of a Lennard-Jones type potential and the angle parts favour a certain geometry for the minimizer (square or hexagonal). Then, we can easily believe that our results (especially Theorem \ref{THintro3} on local minimality) can be applied to such models, once the angle parts is taken in order to favour the square or the triangular lattice. It is however absolutely not clear that our conjecture would stay true for this kind of potential. Furthermore, it has been recently shown in \cite{BetKnuepfspatiallyextended} that the (global and local) minimality results obtained for the lattice theta function stay true for interactions of radially symmetric masses. Therefore, it is natural to think that the result of our paper could stay true for Lennard-Jones interactions of masses which would be sufficiently concentrated around the lattice sites.

\medskip

\textbf{Plan of the paper:} The next section is devoted to the precise definition of the potentials, energy and lattices. In Section \ref{part2}, we compute the two first derivatives of our energy $E_f$ in the general case. Thus, we apply these results to Lennard-Jones type potentials $V_{a,t}^{LJ}$ in Section \ref{part3} and we prove Theorem \ref{THintro3}. In Section \ref{part4}, we study numerically $(x,y)\mapsto E_V(x,y,A)$ in the classical Lennard-Jones case $V(r)=r^{-6}-2r^{-3}$, especially among rhombic and rectangular lattices. Our conjecture is explained and justified in Section \ref{summary}.

\section{Lattices, parametrization and energies}\label{part1}

\subsection{Lattice parametrization and general energy}\label{secparam}
Let $L=\Z u\oplus \Z v\subset \R^2$ be a Bravais lattice. We say that $A>0$ is the area (or covolume) of $L$ if $|u\wedge v|=A$, i.e. the area of its primitive cell is $A$. If $L$ is of area $1/2$, we use the usual parametrization (see Rankin \cite{Rankin} or Montgomery \cite{Mont}) of $L$ by 
$$
(x,y)\in \mathcal{D}=\{(x,y)\in \R^2; 0\leq x \leq 1/2 , y>0 ; x^2+y^2\geq 1\},
$$
where $\mathcal{D}$ is the half fundamental modular domain. It actually corresponds to parametrize $u$ and $v$ by $(x,y)\in \mathcal{D}$ such that
$$
u=\left(\frac{1}{\sqrt{2y}},0  \right)\quad \text{and} \quad v=\left( \frac{x}{\sqrt{2y}},\sqrt{\frac{y}{2} } \right).
$$
Thus, a lattice $L_A$ of given area $A>0$ is uniquely parametrized by $u_A$ and $v_A$ such that
$$
L_A=\Z u_A\oplus \Z v_A:=\Z\left(\frac{\sqrt{A}}{\sqrt{y}},0  \right)\oplus \Z\left( \frac{x\sqrt{A}}{\sqrt{y}},\sqrt{A}\sqrt{y}\right),
$$
with $(x,y)\in \mathcal{D}$. The point $(x,y)$ parametrizes the shape of $L_A$ and $A$ it's inverse density. Furthermore, we have, for any $(m,n)\in \Z^2$,
$$
|mu_A+nv_A|^2=A\left[\frac{1}{y}(m+xn)^2+yn^2  \right],
$$
and all these values are the square of the distances from $(0,0)$ to the points of $L_A$. This function of $(m,n)$ is the quadratic form associated to $L_A$.

\medskip

We recall that the triangular lattice of area $A$ (also called ``hexagonal lattice" or ``Abrikosov lattice" in the context of Superconductivity \cite{Sandier_Serfaty}) is defined, up to rotation, by
\begin{equation}\label{deftriA}
\Lambda_A:=\sqrt{\frac{2A}{\sqrt{3}}}\left[\Z (1,0) \oplus \Z (1/2,\sqrt{3}/2)\right],
\end{equation}
and the square lattice of area $A$ is $\sqrt{A}\Z^2$. In Figure \ref{Parametrization}, we have represented the fundamental domain $\mathcal{D}$. The point $(0,1)$ corresponds to the square lattice $2^{-1/2}\Z^2$ of area $1/2$ and $(1/2,\sqrt{3}/2)$ corresponds to the triangular lattice $\Lambda_{1/2}$ of area $1/2$.

\medskip

\begin{figure}[!h]
\centering
\includegraphics[width=13cm]{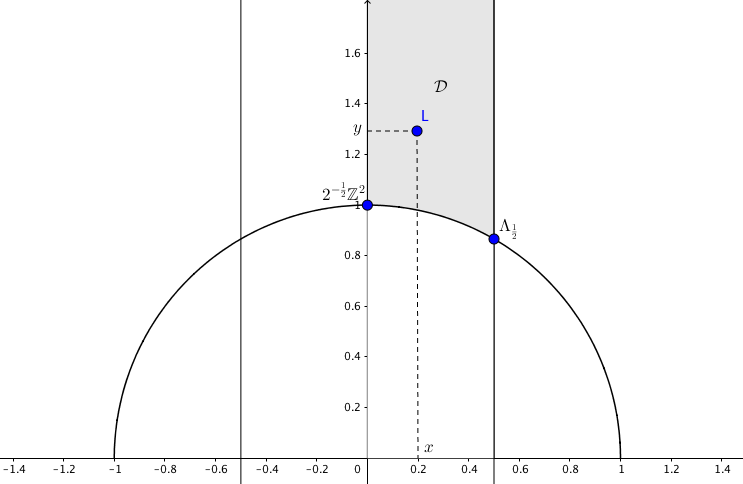}
\caption{Fundamental domain $\mathcal{D}$ and parametrization of a lattice $L$ by $(x,y)$.}
\label{Parametrization}
\end{figure}

We define the space of functions $\mathcal{F}$ by
\begin{equation}\label{deff}
\mathcal{F}:=\left\{f\in C^2(\R_+^*);\forall k\in\{0,1,2\}, |f^{(k)}(r)|=O(r^{-\eta_k-k}),\textnormal{for some } \eta_k>1  \right\}.
\end{equation}
Thus, for any $A>0$, for any Bravais lattice $L_A$ of area $A$ and any $f\in\mathcal{F}$, we define its $f$-energy by
$$
E_f[L_A]=E_f(x,y,A)=\sum_{p\in L_A\backslash\{0\}} f(|p|^2)=\sum_{m,n} f\left(A\left[\frac{1}{y}(m+xn)^2+yn^2  \right]   \right),
$$
where the sum is taken over all $(m,n)\in\Z^2\backslash\{(0,0)\}$. Throughout the paper, we do not specify if the summation is taken over $\Z^2$ or $\Z^2\backslash \{(0,0)\}$ because it will be obvious according to the definition of $f$. Furthermore, the possible value at the origin is the same for any Bravais lattice and then does not have any importance in our problem. Thus, the function $(x,y)\mapsto E_f(x,y,A)$ belongs to $C^2(\mathcal{D})$ and, for any $k\in\{1,2\}$, the $k$-th derivative of $E_f$ is
\begin{align*}
&\partial^{(k)} E_f(x,y,A)=\sum_{m,n}\partial^{(k)}f\left(A\left[\frac{1}{y}(m+xn)^2+yn^2  \right]   \right),
\end{align*}
with respect to any variables. Furthermore, the symmetry $E_f(-x,y,A)=E_f(x,y,A)$ justifies the fact that we study $(x,y)\mapsto E_f(x,y,A)$ in the half modular domain $\mathcal{D}$.

\subsection{Rhombic and rectangular lattices}\label{rhombrectdef}

\begin{defi}[Rhombic lattice]
We say that a Bravais lattice $L_A=\Z u_A \oplus \Z v_A$, parametrized by $(x,y,A)$, is rhombic if it is generated by two vectors of the same length $|u_A|=|v_A|$, which is equivalent to $x^2+y^2=1$. In particular, if $L_A$ is rhombic, then there exists $\theta\in \left[ 60^\circ,90^\circ \right]$ such that $x=\cos\theta$ and $y=\sin\theta$. Thus, we define, for any $f\in\mathcal{F}$, any $60^\circ\leq \theta \leq 90^\circ$ and any $A>0$,
$$
E_f(\theta,A):=E_f(\cos\theta,\sin\theta,A).
$$
\end{defi}

\begin{lemma}If $L_A=\Z u_A \oplus \Z v_A$ is rhombic and $(x,y)=(\cos\theta,\sin\theta)$, then $(\widehat{u_A,v_A})=\theta$.
\end{lemma}
\begin{proof}
This is clear because, since $L$ is rhombic, we have
$$
u_A\cdot v_A=\frac{A x}{y}=|u_A||v_A|\cos(\widehat{u_A,v_A})=A\frac{\sqrt{x^2+y^2}}{y}\cos(\widehat{u_A,v_A})=\frac{A\cos(\widehat{u_A,v_A})}{y}.
$$
Therefore $\cos(\widehat{u_A,v_A})=x=\cos\theta$ and $(\widehat{u_A,v_A})=\theta$ because $\theta\in [60^\circ,90^\circ]$.
\end{proof}

\begin{defi}[Rectangular lattice]
We say that a Bravais lattice $L_A$, parametrized by $(x,y,A)$, is rectangular if its primitive cell is a rectangle, i.e. $u_A\bot v_A$ or, equivalently, if $x=0$ and $y\geq 1$. Thus, we define, for any $f\in\mathcal{F}$, any $y\geq 1$ and any $A>0$,
$$
E_f(y,A):=E_f(0,y,A).
$$
\end{defi}

\begin{remark}
If $L_A$ is rectangular, then it is generated by  $\displaystyle u_A=\sqrt{A}\left(\frac{1}{\sqrt{y}},0\right)$ and $\displaystyle v_A=\sqrt{A}\left(0,\sqrt{y}  \right)$.
\end{remark}

\section{Computation of the first and the second derivatives of $E_f$}\label{part2}
In this part, we compute the first and second derivatives of $(x,y)\mapsto E_f(x,y,A)$ with respect to $x$ and $y$, for fixed $A>0$. We do not give all the details of the computations, but only the key points. 
\subsection{First derivatives}\label{firstderiv}

The following results stay true if there is no condition for the second derivative of $f$. Furthermore, we will find again a result of Coulangeon and Sch\"urmann \cite[Thm. 4.4.(1)]{Coulangeon:2010uq} in the simple two-dimensional case: the square lattice and the triangular lattice are both critical points of $L_A\mapsto E_f[L_A]$ for any $A>0$. Indeed, it turns out (see e.g. \cite{Martinet}) that each shell (or layer) of $\Lambda_A$ and $\sqrt{A}\Z^2$ is a spherical $2$-design, i.e. for any shell $\mathcal{S}$ of $\Lambda_A$ or $\sqrt{A}\Z^2$ that belongs to the circle $C_r$ of radius $r$ we have
$$
\frac{1}{2\pi r}\int_{C_r} p(x)dx=\frac{1}{\sharp \mathcal{S}}\sum_{x\in \mathcal{S}} p(x)
$$
for any polynomial $p$ of degree up to $2$.
 
\begin{prop}\label{deriv}
We have, for any $f\in \mathcal{F}$, any $A>0$ and any $(x,y)\in \mathcal{D}$,
\begin{align*}
&\partial_x E_f(x,y,A)=\frac{2A}{y}\sum_{m,n} (mn+n^2x)f'\left(A\left[\frac{1}{y}(m+xn)^2+yn^2  \right]   \right),\\
&\partial_y E_f(x,y,A)=-\frac{A}{y^2}\sum_{m,n} (m^2+2xmn+(x^2-y^2)n^2)f'\left(A\left[\frac{1}{y}(m+xn)^2+yn^2  \right]   \right).
\end{align*}
\end{prop}

\begin{prop}\label{D1xy}
For any $A>0$ and any $f\in\mathcal{F}$, $(0,1)$ is a critical point of $(x,y)\mapsto E_f(x,y,A)$.
\end{prop}
\begin{proof}
By Proposition \ref{deriv}, we get
\begin{align*}
&\partial_x E_f(0,1,A)=2A\sum_{m,n} mn f'\left(A\left[m^2+n^2 \right]   \right),\\
&\partial_y E_f(0,1,A)=-A\sum_{m,n} (m^2-n^2)f'\left(A\left[m^2+n^2 \right]   \right).
\end{align*}
The first sum is equal to zero by pairing $(m,n)$ and $(-m,n)$. The second is equal to zero because
$$
\sum_{m,n} m^2 f(A[m^2+n^2])=\sum_{m,n} n^2 f(A[m^2+n^2])
$$
by exchange of variables.
\end{proof}

\begin{lemma}\label{sumlatticetri}
For any $(m,n)\in \Z^2\backslash\{(0,0) \}$, let $q(m,n)=m^2+mn+n^2$ and $F:\R\to \R$ be such that the following sums are convergent, then
\begin{equation}\label{LATSUM1}
\sum_{m,n}mn F(q(m,n))=-\frac{1}{2}\sum_{m,n} n^2 F(q(m,n)),
\end{equation}
\begin{equation}\label{LATSUM2}
\sum_{m,n}n^3m F(q(m,n))=-\frac{1}{2}\sum_{m,n}n^4 F(q(m,n)),
\end{equation}
\begin{equation}\label{LATSUM3}
\sum_{m,n}m^2n^2 F(q(m,n))=\frac{1}{2}\sum_{m,n}n^4 F(q(m,n)).
\end{equation}
\end{lemma}
\begin{proof}
The key point is the fact that, for any $(m,n)\in\Z^2\backslash\{(0,0) \}$, 
$$
q(-m-n,n)=q(m,n).
$$
Consequently, we get
\begin{align*}
\sum_{m,n}mnF(q(m,n))&=\sum_{m,n}(-m-n)n F(q(m,n))\\
&=-\sum_{m,n}mnF(q(m,n))-\sum_{m,n}n^2F(q(m,n)),
\end{align*}
and \eqref{LATSUM1} is proved. For the second equality, we compute
\begin{align*}
\sum_{m,n}mn^3F(q(m,n))&=\sum_{m,n}n^3(-m-n)F(q(m,n))\\
&=-\sum_{m,n}mn^3F(q(m,n))-\sum_{m,n}n^4F(q(m,n)),
\end{align*}
and \eqref{LATSUM2} is proved. For the last one, we remark that, using $q(m,n)=q(n,m)$,
\begin{align*}
\sum_{m,n}n^4F(q(m,n))&=\sum_{m,n}(-m-n)^4F(q(m,n))\\
&=\sum_{m,n}(2m^4+6m^2n^2+8m^3n)F(q(m,n)),
\end{align*}
and it follows that
$$
\sum_{m,n}m^2n^2 F(q(m,n))=-\frac{1}{6}\sum_{m,n}n^4 F(q(m,n))-\frac{4}{3}\sum_{m,n}m^3n F(q(m,n)).
$$
Combining this equality with \eqref{LATSUM2}, we get the result.
\end{proof}

\begin{prop}\label{triangCM}
For any $A>0$ and any $f\in \mathcal{F}$, $\left(\frac{1}{2},\frac{\sqrt{3}}{2}   \right)$ is a critical point of $(x,y)\mapsto E_f(x,y,A)$.
\end{prop}
\begin{proof}
Using Proposition \ref{deriv}, we obtain
$$
\partial_x E_f\left( \frac{1}{2},\frac{\sqrt{3}}{2},A  \right)=\frac{4A}{\sqrt{3}}\sum_{m,n}\left(mn+\frac{n^2}{2}  \right)f'\left(\frac{2A}{\sqrt{3}}[m^2+mn+n^2]   \right)
$$
and
$$
\partial_y E_f\left( \frac{1}{2},\frac{\sqrt{3}}{2} ,A \right)=-\frac{2A}{3}\sum_{m,n}\left(m^2+mn-\frac{n^2}{2}\right)f'\left(\frac{2A}{\sqrt{3}}[m^2+mn+n^2]   \right).
$$
We remark, exchanging $m$ and $n$, that
$$
\partial_y E_f\left( \frac{1}{2},\frac{\sqrt{3}}{2},A  \right)=-\frac{2A}{3}\sum_{m,n}\left(mn+\frac{n^2}{2}  \right)f'\left(\frac{2A}{\sqrt{3}}[m^2+mn+n^2]   \right)=-\frac{\partial_x E_f\left( \frac{1}{2},\frac{\sqrt{3}}{2},A  \right)}{2\sqrt{3}}.
$$
Thus, by \eqref{LATSUM1}, we get
$$
\sum_{m,n}mn f'\left(\frac{2A}{\sqrt{3}}[m^2+mn+n^2] \right)=-\frac{1}{2}\sum_{m,n}n^2 f'\left(\frac{2A}{\sqrt{3}}[m^2+mn+n^2] \right),
$$
i.e.
$$
\sum_{m,n}\left(mn+\frac{n^2}{2}  \right)f'\left(\frac{2A}{\sqrt{3}}[m^2+mn+n^2]   \right)=0,
$$
and the result is proved.
\end{proof}

Now we recall a simple application of Montgomery results \cite{Mont} to the case of completely monotone interacting potentials. 

\begin{defi}\label{defCMfct}
We say that $f:(0,\infty)\to \R$ is completely monotone if, for any $k\in \N$ and any $r>0$, $(-1)^k f^{(k)}(r)\geq 0$.
\end{defi}

\begin{prop}\label{genmgt} (\cite{Mont})
If $f\in \mathcal{F}$ is completely monotone, then for any $A>0$ and for any $(x,y)$ such that $0<x<1/2$ and $y>\sqrt{3}/2$, we have 
$$
\partial_x E_f(x,y,A)<0 \quad \textnormal{and} \quad  \partial_y E_f(x,y,A)>0.
$$ 
In particular, $(x,y)=(1/2,\sqrt{3}/2)$ is the only minimizer of $(x,y)\mapsto  E_f(x,y,A)$ and $x=(0,1)$ is a saddle point. Furthermore, this function has no other critical point.
\end{prop}
\begin{proof}
It is clear by Montgomery results \cite[Lem. 4 and 7]{Mont} and the fact (see e.g. \cite[Section 3]{BetTheta15}) that any completely monotone function $f$ can be written as the Laplace transform of a positive Borel measure $\mu$ on $[0,+\infty)$, i.e.
$$
f(r)=\int_0^{+\infty} e^{-rt}d\mu(t).
$$
\end{proof}

\begin{examples}
In particular, the previous proposition holds for the Epstein zeta function and the theta functions defined by \eqref{defEpsttheta}.
\end{examples}

\subsection{Second derivatives}

\begin{prop}
For any $A>0$, any $f\in\mathcal{F}$ and any $(x,y)\in\mathcal{D}$, we have
\begin{align*}
\partial^2_{xx}E_f(x,y,A)=&\frac{2A}{y}\sum_{m,n}n^2 f'\left(A\left[\frac{1}{y}(m+xn)^2+yn^2  \right]   \right)\\
&+\frac{4A^2}{y^2}\sum_{m,n}(mn+n^2x)^2f''\left(A\left[\frac{1}{y}(m+xn)^2+yn^2  \right]   \right),
\end{align*}

\begin{align*}
\partial^2_{yy}E_f(x,y,A)=&\frac{2A}{y^3}\sum_{m,n}(m+xn)^2f'\left(A\left[\frac{1}{y}(m+xn)^2+yn^2  \right]   \right)\\
&+A^2\sum_{m,n}\left(n^2-\frac{(m+xn)^2}{y^2}  \right)^2f''\left(A\left[\frac{1}{y}(m+xn)^2+yn^2  \right]   \right),
\end{align*}
and

\begin{align*}
\partial^2_{xy}E_f(x,y,A)=&-\frac{2A}{y^2}\sum_{m,n}(mn+n^2x)f'\left(A\left[\frac{1}{y}(m+xn)^2+yn^2  \right]   \right)\\
&+\frac{2A^2}{y}\sum_{m,n}(mn+n^2x)\left( n^2-\frac{(m+xn)^2}{y^2}  \right)f''\left(A\left[\frac{1}{y}(m+xn)^2+yn^2  \right]   \right).
\end{align*}
In particular, if $(x,y)\in\mathcal{D}$ is a critical point of $(x,y)\mapsto E_f(x,y,A)$, then
\begin{equation*}
\partial^2_{xy}E_f(x,y,A)=\frac{2A^2}{y}\sum_{m,n}(mn+n^2x)\left( n^2-\frac{(m+xn)^2}{y^2}  \right)f''\left(A\left[\frac{1}{y}(m+xn)^2+yn^2  \right]   \right).
\end{equation*}
\end{prop}
\begin{proof}
It is a straightforward computation. The last point follows from the equation $\partial_x E_f(x,y,A)=0$ and the expression of $\partial^2_{xy}E_f(x,y,A)$.
\end{proof}

\begin{corollary}\label{Hessienne01}
Let $A>0$ and $f\in\mathcal{F}$, then the second derivatives of $(x,y)\mapsto E_f(x,y,A)$ at point $(0,1)$ are:
\begin{align*}
&\partial^2_{xx}E_f(0,1,A)=2A\sum_{m,n}n^2 f'\left(A\left[m^2+n^2 \right]   \right)+4A^2\sum_{m,n}m^2n^2 f''\left(A\left[m^2+n^2 \right]   \right),\\
& \partial^2_{yy}E_f(0,1,A)=2A\sum_{m,n}m^2 f'\left(A\left[m^2+n^2 \right]   \right)+A^2\sum_{m,n}(n^2-m^2)^2f''\left(A\left[m^2+n^2 \right]   \right),\\
&\partial^2_{xy}E_f(0,1,A)=0.
\end{align*}
\end{corollary}
\begin{proof}
The both first results are obvious. Furthermore, we have
\begin{align*}
&\partial^2_{xy}E_f(0,1,A)\\
&=-2A\sum_{m,n}mnf'\left(A\left[m^2+n^2 \right]   \right)+2A\sum_{m,n}mn(n^2-m^2)f''\left(A\left[m^2+n^2 \right]   \right)=0
\end{align*}
by pairing $(m,n)$ and $(-m,n)$ in each sum.
\end{proof}

\begin{prop}\label{CNminsquare}
If $A>0$ and $f\in \mathcal{F}$ are such that
\begin{align*}
&K_f^1(A):=\sum_{m,n}n^2 f'\left(A\left[m^2+n^2 \right]   \right)+2A\sum_{m,n}m^2n^2 f''\left(A\left[m^2+n^2 \right]   \right)>0,\\
&K_f^2(A):=\sum_{m,n}2m^2 f'\left(A\left[m^2+n^2 \right]   \right)+A\sum_{m,n}(n^2-m^2)^2f''\left(A\left[m^2+n^2 \right]   \right)>0,
\end{align*}
then $(0,1)$ is a local minimizer of $(x,y)\mapsto E_f(x,y,A)$.
\end{prop}
\begin{proof}
It follows directly from Corollary \ref{Hessienne01} and Proposition \ref{D1xy}.
\end{proof}

\begin{prop}\label{derivtriangular}
Let $f\in\mathcal{F}$, then the second derivatives at point $\left(\frac{1}{2},\frac{\sqrt{3}}{2}  \right)$ are:
\begin{align*}
T_f(A):&=\partial^2_{xx}E_f\left(\frac{1}{2},\frac{\sqrt{3}}{2},A  \right)=\partial^2_{yy}E_f\left(\frac{1}{2},\frac{\sqrt{3}}{2},A  \right)\\
&=\frac{4A}{\sqrt{3}}\sum_{m,n}n^2 f'\left(\frac{2A}{\sqrt{3}}[m^2+mn+n^2]  \right)+\frac{4A^2}{3}\sum_{m,n}n^4f''\left(\frac{2A}{\sqrt{3}}[m^2+mn+n^2]  \right),
\end{align*}

and $\displaystyle \partial^2_{xy}E_f\left(\frac{1}{2},\frac{\sqrt{3}}{2},A  \right)=0$.
\end{prop}
\begin{remark}
By Proposition \ref{D1xy}, we already know that the triangular lattice is a critical point of the energy. Consequently, if $T_f(A)>0$, then $\left(\frac{1}{2},\frac{\sqrt{3}}{2}  \right)$ is a local minimizer of $(x,y)\mapsto E_f(x,y,A)$ and if $T_f(A)<0$, then $\left(\frac{1}{2},\frac{\sqrt{3}}{2}  \right)$ is a local maximizer of $(x,y)\mapsto E_f(x,y,A)$.
\end{remark}
\begin{proof}
We define, for any $(m,n)\in\Z^2\backslash\{(0,0)\}$ and any $A>0$, the quadratic form associated to $\Lambda_A$ by
$$
Q_A(m,n):=\frac{2A}{\sqrt{3}}[m^2+mn+n^2]. 
$$
It is straightforward to get
\begin{align*}
\partial^2_{xx}E_f\left(\frac{1}{2},\frac{\sqrt{3}}{2},A  \right)=&\frac{4A}{\sqrt{3}}\sum_{m,n}n^2 f'\left(Q_A(m,n)  \right)+\frac{16A^2}{3}\sum_{m,n}\left( mn+\frac{n^2}{2}\right)^2f''\left(Q_A(m,n)  \right),
\end{align*}
\begin{align*}
\partial^2_{yy}E_f\left(\frac{1}{2},\frac{\sqrt{3}}{2},A  \right)=&\frac{16A}{3\sqrt{3}}\sum_{m,n}\left( m+\frac{n}{2}\right)^2f'\left(Q_A(m,n) \right)\\
&+A^2\sum_{m,n}\left(n^2-\frac{4}{3}\left( m+\frac{n}{2} \right)^2  \right)^2f''\left(Q_A(m,n) \right),
\end{align*}
and
\begin{align*}
\partial^2_{xy}E_f\left(\frac{1}{2},\frac{\sqrt{3}}{2},A  \right)=&-\frac{8A}{3}\sum_{m,n}\left( mn+\frac{n^2}{2} \right)f'\left(Q_A(m,n)   \right)\\
&+\frac{4A^2}{\sqrt{3}}\sum_{m,n}\left( mn+\frac{n^2}{2} \right)\left(n^2-\frac{4}{3}\left( m+\frac{n}{2} \right)^2  \right)f''\left(Q_A(m,n)  \right).
\end{align*}
Now, let us prove that $\partial^2_{xx}E_f\left(\frac{1}{2},\frac{\sqrt{3}}{2},A  \right)=\partial^2_{yy}E_f\left(\frac{1}{2},\frac{\sqrt{3}}{2},A  \right)$, and more precisely that
\begin{equation}\label{latsum1}
\frac{4}{\sqrt{3}}\sum_{m,n}n^2 f'\left(Q_A(m,n) \right)=\frac{16}{3\sqrt{3}}\sum_{m,n}\left( m+\frac{n}{2}\right)^2f'\left(Q_A(m,n) \right)
\end{equation}
and
\begin{equation}\label{latsum2}
\frac{16}{3}\sum_{m,n}\left( mn+\frac{n^2}{2}\right)^2f''\left(Q_A(m,n)  \right)=\sum_{m,n}\left(n^2-\frac{4}{3}\left( m+\frac{n}{2} \right)^2  \right)^2f''\left( Q_A(m,n) \right).
\end{equation}

By \eqref{LATSUM1}, we get, 
\begin{align*}
\frac{16}{3\sqrt{3}}\sum_{m,n}\left( m+\frac{n}{2}\right)^2f'\left(Q_A(m,n) \right)&=\frac{16}{3\sqrt{3}}\sum_{m,n}\left( m^2+\frac{n^2}{4}+mn\right)f'\left(Q_A(m,n) \right)\\
&=\frac{16}{3\sqrt{3}}\sum_{m,n}\left( m^2+\frac{n^2}{4}-\frac{n^2}{2}\right)f'\left(Q_A(m,n) \right)\\
&=\frac{4}{\sqrt{3}}\sum_{m,n}n^2 f'\left(Q_A(m,n) \right),
\end{align*}
and \eqref{latsum1} is proved. For the second equality, applying \eqref{LATSUM2} and \eqref{LATSUM3}, we obtain
\begin{align*}
\sum_{m,n}\left(n^2-\frac{4}{3}\left( m+\frac{n}{2} \right)^2  \right)^2f''\left( Q_A(m,n) \right)&=\frac{4}{9}\sum_{m,n}(5n^4+4m^3n)f''\left( Q_A(m,n) \right)\\
&=\frac{4}{3}\sum_{m,n}n^4 f''\left( Q_A(m,n) \right)
\end{align*}
and
\begin{align*}
\frac{16}{3}\sum_{m,n}\left( mn+\frac{n^2}{2}\right)^2f''\left(Q_A(m,n)  \right)&=\frac{16}{3}\sum_{m,n}\left( m^2n^2+\frac{n^4}{4}+mn^3 \right)f''\left( Q_A(m,n) \right)\\
&=\frac{16}{3}\sum_{m,n}\left(\frac{n^4}{2}+\frac{n^4}{4}-\frac{n^4}{2}   \right)f''\left( Q_A(m,n) \right)\\
&=\frac{4}{3}\sum_{m,n}n^4 f''\left( Q_A(m,n) \right).
\end{align*}
Hence, \eqref{latsum2} is established. By \eqref{LATSUM1}, the first sum in the expression of $\partial^2_{xy}E_f\left(\frac{1}{2},\frac{\sqrt{3}}{2},A  \right)$ is equal to $0$. Combining \eqref{LATSUM2} and \eqref{LATSUM3}, we easily prove that the second part of $\partial^2_{xy}E_f\left(\frac{1}{2},\frac{\sqrt{3}}{2},A  \right)$ is also equal to $0$. 
\end{proof}

\begin{corollary}\label{coranalytic}
If $f\in \mathcal{F}$ is analytic on an open neighbourhood of $(0,+\infty)$, then for almost every $A>0$, $(1/2,\sqrt{3}/2)$ is a local minimizer or a local maximizer of $(x,y)\mapsto E_f(x,y,A)$.
\end{corollary}
\begin{proof}
If $f$ is analytic, then $f'$ and $f''$ are analytic on an open neighbourhood of $(0,+\infty)$ and $T_f$ is also analytic on an open neighbourhood of $(0,+\infty)$. Then, the set of zeros of $A\mapsto T_f(A)$ is a discrete set and $T_f(A)\neq 0$ for almost every $A>0$.
\end{proof}

\begin{examples}
This result is true for any sum of inverse power laws $f(r)=\sum_{i=1}^p a_i r^{-s_i}$, $s_i>1$, any sum of exponential functions or any sum of those types of functions (see \cite{BetTheta15} for more examples).
\end{examples}

\section{Application to Lennard-Jones type interactions}\label{part3}
The aim of this part is to apply the previous results to the class of Lennard-Jones type potentials (also called ``Mie potentials"). We recall our definition from \cite[Section 6.3]{BetTheta15}.

\begin{defi}\label{defiLJpoten}
For any $t=(t_1,t_2)\in \R^2$ such that $1<t_1<t_2$ and any $a=(a_1,a_2)\in (0,+\infty)^2$, we define the Lennard-Jones type potential (see Figure \ref{Ljonespot} for its graph) on $(0,+\infty)$ by
$$
V_{a,t}^{LJ}(r):=\frac{a_2}{r^{t_2}}-\frac{a_1}{r^{t_1}}.
$$
Hence, its lattice energy is defined, for any Bravais lattice $L\subset \R^2$, by
$$
E_{V_{a,t}^{LJ}}[L]=a_2\zeta_L(2t_2)-a_1\zeta_L(2t_1),
$$
where the Epstein zeta function $\zeta_L$ is defined by \eqref{defEpsttheta}. Furthermore, we define the following lattice sums:

\begin{minipage}[l]{0.4\linewidth}
\begin{align*}
&S_1(s)=\sum_{m,n}\frac{m^4}{(m^2+mn+n^2)^s},\\
&S_3(s)=\sum_{m,n}\frac{m^2n^2}{(m^2+n^2)^s},
\end{align*}
\end{minipage}
\begin{minipage}[c]{0.5\linewidth}
\begin{align*}
&S_2(s)=\sum_{m,n}\frac{m^2}{(m^2+n^2)^s},\\
&S_4(s)=\sum_{m,n}\frac{(n^2-m^2)^2}{(m^2+n^2)^s}.
\end{align*}
\end{minipage}
\begin{figure}[!h]
\centering
\includegraphics[width=10cm,height=80mm]{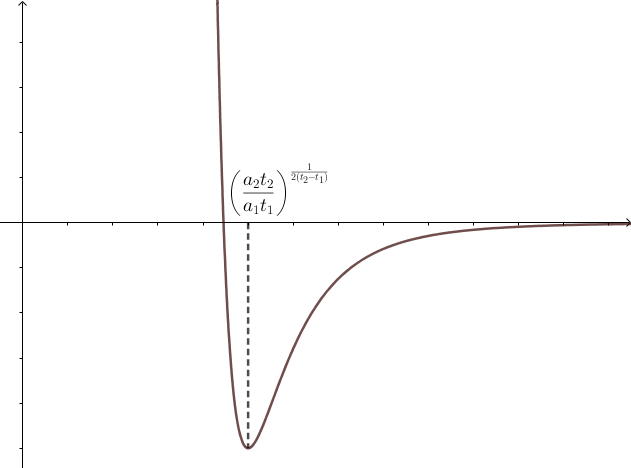}
\caption{Graph of $r\mapsto V_{a,t}^{LJ}(r^2)$}
\label{Ljonespot}
\end{figure}
\end{defi}

As we explained in \cite[Section 6.3]{BetTheta15}, these Lennard-Jones potentials are used in molecular simulation (classical interaction between atoms, hydrogen bonds, for finding energetically favourable regions in protein binding sites) or in the study of social aggregation \cite{MEKBS}. In particular, the classical $(12-6)$ Lennard-Jones potential (see \cite{LJ}) is a good simple model that approximates the interaction between neutral atoms.

\begin{thm}\label{THmain1}
For any $(a,t)$ as in Definition \ref{defiLJpoten}, let 
$$
A_{0}:=\frac{\sqrt{3}}{2}\left( \frac{a_2 t_2(t_2-1)S_1(t_2+2)}{a_1t_1(t_1-1)S_1(t_1+2)} \right)^{\frac{1}{t_2-t_1}},
$$
then we have:
\begin{enumerate}
\item if $A<A_0$, then $\left(\frac{1}{2},\frac{\sqrt{3}}{2}  \right)$ is a local minimizer of $(x,y)\mapsto E_{V_{a,t}^{LJ}}(x,y,A)$;
\item if $A>A_0$, then $\left(\frac{1}{2},\frac{\sqrt{3}}{2}  \right)$ is a local maximizer of $(x,y)\mapsto E_{V_{a,t}^{LJ}}(x,y,A)$.
\end{enumerate}
\end{thm}
\begin{proof}
According to Proposition \ref{derivtriangular}, we easily get
\begin{align*}
T_f(A)=\frac{4A}{\sqrt{3}}\left( \frac{\sqrt{3}}{2A} \right)^{t_2+1}\left\{-a_1t_1\left(\sqrt{3}/2  \right)^{t_1-t_2}h(t_1)A^{t_2-t_1} +a_2t_2h(t_2)\right\},
\end{align*}
where $\displaystyle h(s)=\frac{s+1}{2}S_1(s+2)-\sum_{m,n}\frac{m^2}{(m^2+mn+n^2)^{s+1}}$. Combining \eqref{LATSUM2} and \eqref{LATSUM3}, we remark that
\begin{align*}
\sum_{m,n}\frac{m^2}{(m^2+mn+n^2)^{s+1}}&=\sum_{m,n}\frac{m^2(m^2+mn+n^2)}{(m^2+mn+n^2)^{s+2}}\\
&=\sum_{m,n}\frac{m^2n^2+mn^3+m^4}{(m^2+mn+n^2)^{s+2}}\\
&=S_1(s+2).
\end{align*}
Consequently, we obtain
$$
h(s)=\frac{s-1}{2}S_1(s+2),
$$
and $T_f(A)>0$ if and only if $A<A_0$. The second point is clear.
\end{proof}

\begin{remark}
In the particular classical case $a=(2,1)$ and $t=(3,6)$, the interaction potential is
$$
V(r)=\frac{1}{r^6}-\frac{2}{r^3},
$$
and $\displaystyle V(r^2)=\frac{1}{r^{12}}-\frac{2}{r^6}$ is the so-called classical Lennard-Jones potential. The previous result shows that, for any $A<A_0\approx 1.152$ (resp. $A>A_0$) the triangular lattice is a local minimizer  (resp. maximizer) of $L_A\mapsto E_{V}[L_A]$.
\end{remark}

\begin{thm}\label{THmain2}
For any $(a,t)$ as in Definition \ref{defiLJpoten}, let
$$
g(s)=S_2(s+1)-2(s+1)S_3(s+2),\quad k(s)=(s+1)S_4(s+2)-2S_2(s+1),
$$
and define
$$A_1:=\left(\frac{a_2t_2g(t_2)}{a_1t_1g(t_1)}  \right)^{\frac{1}{t_2-t_1}} \quad \textnormal{and}\quad A_2:=\left(\frac{a_2t_2k(t_2)}{a_1t_1k(t_1)}  \right)^{\frac{1}{t_2-t_1}}.
$$
It holds:
\begin{enumerate}
\item if $A_1<A<A_2$, then $(0,1)$ is a local minimizer of $(x,y) \mapsto E_{V_{a,t}^{LJ}}(x,y,A)$;
\item if $A\not\in [A_1,A_2]$, then $(0,1)$ is a saddle point of $(x,y) \mapsto E_{V_{a,t}^{LJ}}(x,y,A)$.
\end{enumerate}
\end{thm}

\begin{proof}
We apply Proposition \ref{CNminsquare} and we compute
\begin{align*}
&K_{V_{a,t}^{LJ}}^1(A)=\frac{1}{A^{t_2+1}}\left\{a_1t_1 g(t_1)A^{t_2-t_1}-a_2t_2g(t_2) \right\},\\
&K_{V_{a,t}^{LJ}}^2(A)=\frac{1}{A^{t_2+1}}\left\{-a_1t_1 k(t_1)A^{t_2-t_1}+a_2t_2k(t_2) \right\}.
\end{align*}
We now remark that $g(s)>0$ and $k(s)>0$. Indeed, we have
\begin{align*}
g(s)&=\sum_{m,n}\frac{m^2}{(m^2+n^2)^{s+1}}-2(s+1)\sum_{m,n}\frac{m^2n^2}{(m^2+n^2)^{s+2}}\\
&=\sum_{m,n}\frac{m^2(m^2+n^2)}{(m^2+n^2)^{s+2}}-2(s+1)\sum_{m,n}\frac{m^2n^2}{(m^2+n^2)^{s+2}}\\
&=\sum_{m,n}\frac{m^4}{(m^2+n^2)^{s+2}}-(2s+1)\sum_{m,n}\frac{m^2n^2}{(m^2+n^2)^{s+2}}.
\end{align*}
By change of variable $(m,n)=(k+\ell,k-\ell)$, we obtain, since the number of terms in the right-hand sum is larger than in the left-hand one,
\begin{align*}
\sum_{m,n}\frac{m^2n^2}{(m^2+n^2)^{s+2}}&\leq \sum_{k,\ell}\frac{(k+\ell)^2(k-\ell)^2}{(2k^2+2\ell^2)^{s+2}}\\
&=\frac{1}{2^{s+2}}\sum_{k,\ell}\frac{k^4+\ell^4-2k^2\ell^2}{(k^2+\ell^2)^{s+2}},\\
&=\frac{1}{2^{s+1}}\sum_{k,\ell}\frac{k^4}{(k^2+\ell^2)^{s+2}}-\frac{1}{2^{s+1}}\sum_{k,\ell}\frac{k^2\ell^2}{(k^2+\ell^2)^{s+2}},
\end{align*}
i.e. $\displaystyle \sum_{m,n}\frac{m^2n^2}{(m^2+n^2)^{s+2}}\leq \frac{1}{1+2^{s+1}}\sum_{m,n}\frac{m^4}{(m^2+n^2)^{s+2}}$. Therefore, we get, for any $s>1$,
$$
g(s)\geq \left( 1-\frac{1+2s}{1+2^{s+1}} \right)>0.
$$
With exactly the same arguments, it follow that, for any $s>1$,
$$
k(s)\geq \left(\frac{s2^{s+2}-4}{1+2^{s+1}}  \right)\sum_{m,n}\frac{m^4}{(m^2+n^2)^{s+2}}>0.
$$
Hence, the result is proved because
$$
K_{V_{a,t}^{LJ}}^1(A)>0\iff A>A_1\quad \textnormal{and} \quad K_{V_{a,t}^{LJ}}^2(A)>0\iff A<A_2.
$$
\end{proof}

\begin{remark}
In the classical Lennard-Jones case $a=(2,1)$ and $t=(3,6)$, we numerically compute $A_1\approx 1.143$ and $A_2\approx 1.268$. In particular, if $A>A_2$, then the square lattice cannot be a minimizer of $(x,y)\mapsto E_{V}(x,y,A)$.
\end{remark}

\section{The classical Lennard-Jones energy: numerical study, degeneracy as $A\to +\infty$ and conjecture}\label{part4}

In this part, we study the energy per point associated to the classical Lennard-Jones potential, i.e. $a=(2,1)$ and $t=(3,6)$. The corresponding interaction potential is hence given by
$$
V(r)=\frac{1}{r^6}-\frac{2}{r^3},
$$
and its lattice energy is defined, for any Bravais lattice $L$, by
$$
E_V[L]=\zeta_L(12)-2\zeta_L(6).
$$
\subsection{Minimality among rhombic lattices}\label{secrhomb}

In Table \ref{table2}, we give the results of our numerical and theoretical investigations for the minimization of 
$$
\theta\mapsto E_V(\theta,A):=E_V(\cos\theta,\sin\theta,A)
$$
with respect to the area $A$. For any fixed $A>0$, we call $\theta_A$ a minimizer of $\theta\mapsto E_V(\theta,A)$. We have split $(0,+\infty)$ into four domains Rhi, $1\leq i\leq 4$, and we explain below the results we obtained.

\begin{table}[!h]
\centering
\begin{tabular}{|c|c|c|c|}
\hline
\textbf{Domain}  & \textbf{Area $A$} & \textbf{Minimizer $\theta_A$} & \textbf{Status}\\
\hline
Rh1 & $0<A<\frac{\pi}{(120)^{1/3}}\approx 0.637$ & $60^\circ$ & proved in \cite{Betermin:2014fy}\\
\hline
Rh2 & $\frac{\pi}{(120)^{1/3}}<A<A_{BZ}\approx 1.138$ & $60^\circ$ & num.+loc. min. proved in Th. \ref{THmain1}\\
\hline
Rh3 & $A_{BZ}<A<A_1\approx 1.143$ & $76.43^\circ\leq \theta< 90^\circ$ & num.\\
\hline
Rh4 & $A>A_1$ & $90^\circ$ & num.+loc. min. proved in Th. \ref{THmain2}\\
\hline
\end{tabular}
\caption{Summary of our numerical and theoretical studies for the minimization among rhombic lattices of $\theta \mapsto E_V(\theta,A)$ with respect to the value of $A$.}
\label{table2}
\end{table}

In \cite[Theorem 3.1]{Betermin:2014fy}, we proved that if $0<A<\frac{\pi}{(120)^{1/3}}$, then $\Lambda_A$ is the unique minimizer of $L\mapsto E_V[L]$ among Bravais lattices of fixed area $A$. Hence, it is clear that, on Rh1, $\theta_A=60^\circ$.

\medskip

The optimality of $\theta_A=60^\circ$ on Rh2 follows from \cite[Proposition 3.5]{Betermin:2014fy}. Indeed, we proved that $\Lambda_A$ is a minimizer of $L\mapsto E_V[L]$ among Bravais lattices of fixed area $A$ if and only if 
$$
A\leq A_{BZ}:=\inf_{|L|=1\atop L\neq \Lambda_1} \left(\frac{\zeta_L(12)-\zeta_{\Lambda_1}(12)}{2(\zeta_L(6)-\zeta_{\Lambda_1}(6))}  \right)^{1/3},
$$ 
and we numerically compute that $\displaystyle A_{BZ}\approx 1.138$. Furthermore, we have plotted in Figure  \ref{Fig1} $\theta\mapsto E_V(\theta,1)$. It turns out that this $A=1$ case plays an important role for the global minimization (i.e. without a density constraint) of $L\mapsto E_V[L]$ (see Section \ref{globminrmk}).

\begin{figure}[!h]
\centering
\includegraphics[width=8cm]{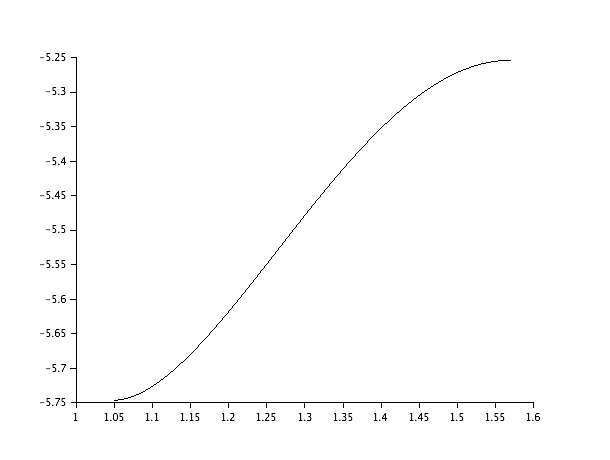}
\caption{Plot of $\theta\mapsto E_V(\theta,1)$, for $A=1$, on $[60^\circ,90^\circ]$ (in radians on the figure). The minimizer seems to be $\theta_A=60^\circ$.}
\label{Fig1}
\end{figure}

For an area between $A_{BZ}\approx 1.138$ and $A_1\approx 1.143$, the minimizer seems (numerically) to cover monotonically and continuously the interval $[76.43^\circ,90^\circ)$ (see Figure \ref{Fig2}, \ref{Fig3} and \ref{Fig4}). There is no doubt about the fact that the transition from $60^\circ$ to $76.43^\circ$ is discontinuous (see Figure \ref{Fig2}).

\begin{figure}[!h]
\centering
\includegraphics[width=8cm, trim=0cm 0cm 1mm 0cm, clip]{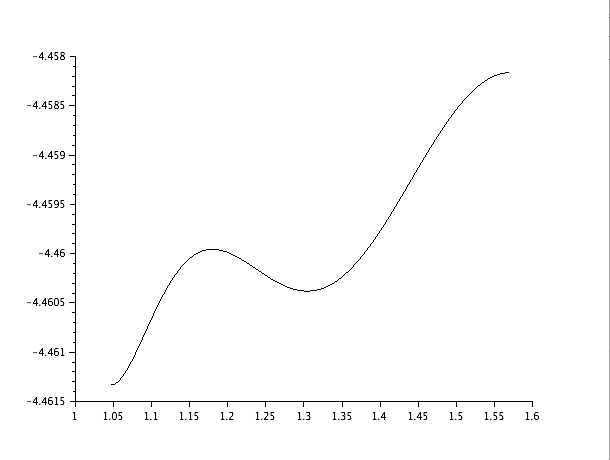} \includegraphics[width=8cm]{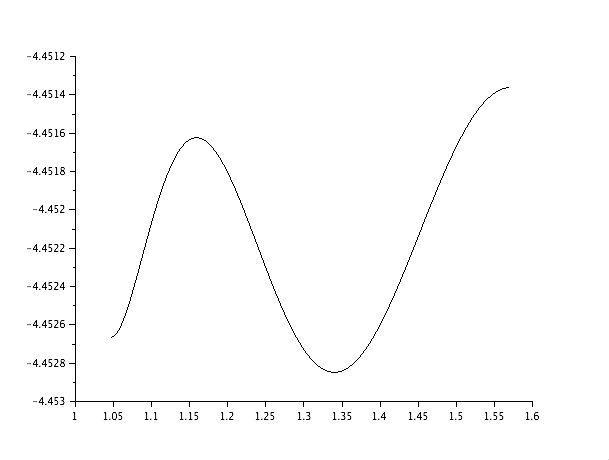}
\caption{Plots of $\theta\mapsto E_V(\theta,A)$ on $[60^\circ,90^\circ]$ (in radians on the figure), for $A=1.137$ (on the left) and $A=1.138$ (on the right). The minimizer seems to be $\theta_A=60^\circ$ for $A=1.137$ and $\theta_A=76.43^\circ$ for $A=1.138$.}
\label{Fig2}
\end{figure}

\begin{figure}[!h]
\centering
\includegraphics[width=8cm]{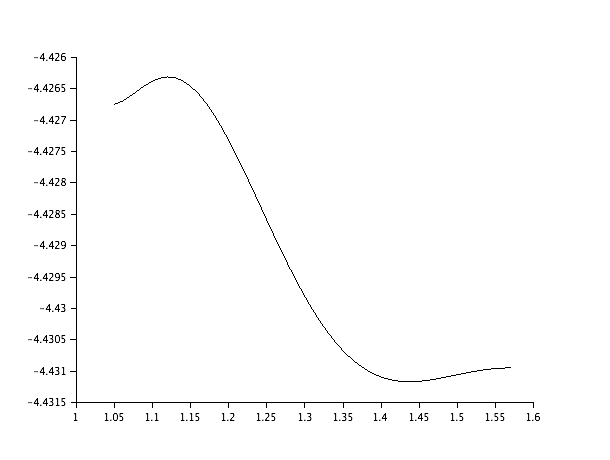}
\caption{Plot of $\theta\mapsto E_V(\theta,1.141)$ on $[60^\circ,90^\circ]$ (in radians on the figure), for $A=1.141$. The minimizer seems to be $\theta_A\approx 82.51^\circ$.}
\label{Fig3}
\end{figure}

\begin{figure}[!h]
\centering
\includegraphics[width=8cm]{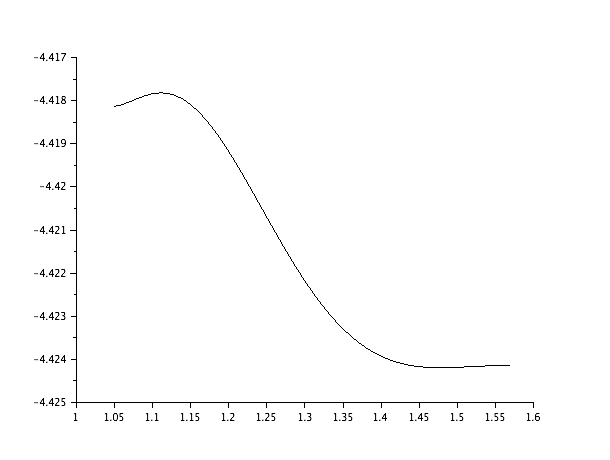} \includegraphics[width=8cm]{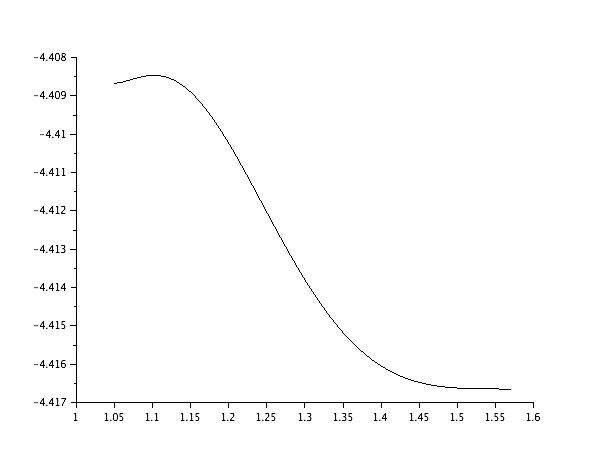}
\caption{Plots of $\theta\mapsto E_V(\theta,A)$ on $[60^\circ,90^\circ]$ (in radians on the figure), for $A=1.142$ (on the left) and $A=1.1431$ (on the right). The minimizer seems to be $\theta_A\approx 89.74^\circ$ for $A=1.142$ and $\theta_A=90^\circ$ for $A=1.1431$.}
\label{Fig4}
\end{figure}

For $A$ in the domain Rh4, our numerical simulations give us the optimality of $\theta_A=90^\circ$ for $A_1<A<20$. We will see in the next subsection that the minimizer stays rectangular if $A$ is large enough (see Figure \ref{Fig5} for the $A=3$ case).

\begin{figure}[!h]
\centering
\includegraphics[width=8cm]{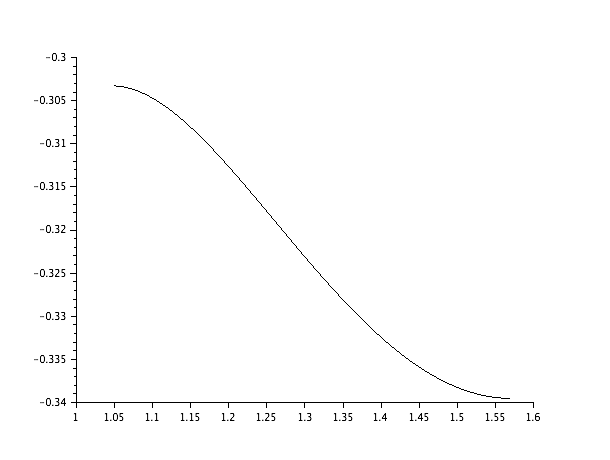}
\caption{Plot of $\theta\mapsto E_V(\theta,3)$ on $[60^\circ,90^\circ]$ (in radians on the figure), for $A=3$. The minimizer seems to be $\theta_A=90^\circ$.}
\label{Fig5}
\end{figure}

\begin{remark}
It numerically appears that the minimizers of $(x,y)\mapsto E_V(x,y,A)$ on $\mathcal{D}$ are rhombic lattices if $0<A<A_2$.
\end{remark}

\subsection{Minimality among rectangular lattices}\label{rectangular}

As in the previous subsection, we give the results of our numerical and theoretical investigations for the minimization of 
$$
y\mapsto E_V(y,A):=E_V(0,y,A)\quad \textnormal{on}\quad [1,+\infty)
$$ 
with respect to area $A$ in Table \ref{table3}. For any fixed $A>0$, we call $y_A$ a minimizer of $y\mapsto E_V(y,A)$. We have split $(0,+\infty)$ into three domains Recti, $1\leq i\leq 3$ and we explain below the results we obtained. In particular, we will partially explain the behavior of the minimizer on Rect3.

\begin{table}[!h]
\centering
\begin{tabular}{|c|c|c|c|}
\hline
\textbf{Domain}  & \textbf{Area $A$} & \textbf{Minimizer $y_A$}& \textbf{Status}\\
\hline
Rect1 & $0<A<\frac{\pi}{(120)^{1/3}}\approx 0.637$ & $1$ & proved in \cite{Betermin:2014fy}\\
\hline
Rect2 & $\frac{\pi}{(120)^{1/3}}<A<A_2\approx 1.268$ & $1$ & num.+loc. min. proved in Th. \ref{THmain2}\\
\hline
Rect3 & $A>A_2$ & $\nearrow$ from $1$ to $+\infty$ & num.+proved for large A in Prop. \ref{degrect} \\
\hline
\end{tabular}
\caption{Summary of our numerical and theoretical studies for the minimization among rhombic lattices of $y\mapsto E_V(y,A)$.}
\label{table3}
\end{table}

The optimality of $y_A=1$, i.e. the square lattice, on Rect1 is clear by \cite[Theorem 3.1]{Betermin:2014fy} and Montgomery result \cite[Lemma 7]{Mont}. Indeed, Montgomery proved that $\partial_y \theta(x,y,\alpha)\geq 0$ for any $(x,y)\in\mathcal{D}$ and any $\alpha>0$, where $\theta(x,y,\alpha):=E_{f_\alpha}(x,y,1/2)$ and $f_\alpha(r)=e^{-\pi\alpha r}$. Furthermore, we proved in \cite[Theorem 3.1]{Betermin:2014fy} that, for any $0<A<\frac{\pi}{(120)^{1/3}}$ and any Bravais lattice $L_A$ with area $A$,
$$
E_V[L_A]=C_A+\frac{\pi^3}{A^3}\int_1^{+\infty}\left(\theta_{L_A}\left(\frac{\alpha}{2A}  \right)-1  \right)g_A(\alpha)\frac{d\alpha}{\alpha},
$$
where $C_A$ is a constant depending on $A$ but independent of $L_A$, and $g_A(\alpha)\geq 0$ for any $\alpha\geq 1$. Thus, we get
$$
\partial_y E_V(y,A)\geq 0
$$ 
for any $y\geq 1$ when $0<A<\frac{\pi}{(120)^{1/3}}$. Therefore, $y=1$ is the unique minimizer of $y\mapsto E_V(y,A)$.

\medskip

On Rect2, $y=1$ seems numerically to be the minimizer (see Figure \ref{Fig6}). Actually, it is not difficult to prove rigorously, by using the algorithmic method detailed in \cite[Lem. 4.19]{BeterminPetrache}, that $y\mapsto E_V(y,A)$ is an increasing function, for any chosen value $A\in\textnormal{Rect2}$.

\begin{figure}[!h]
\centering
\includegraphics[width=8cm]{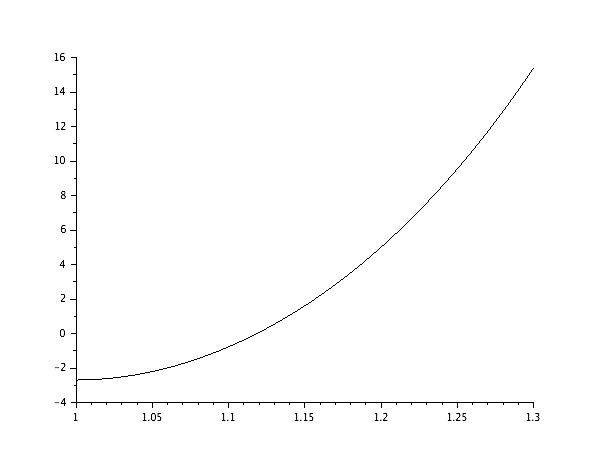} \includegraphics[width=8cm]{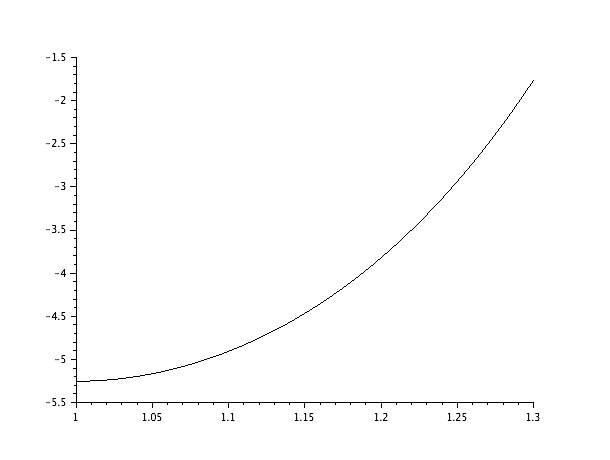}
\caption{Plots of $y\mapsto E_V(y,A)$ for $A=0.8$ (on the left) and $A=1$ (on the right). It seems that the minimizer is $y_A=1$.}
\label{Fig6}
\end{figure}

\begin{figure}[!h]
\centering
\includegraphics[width=8cm]{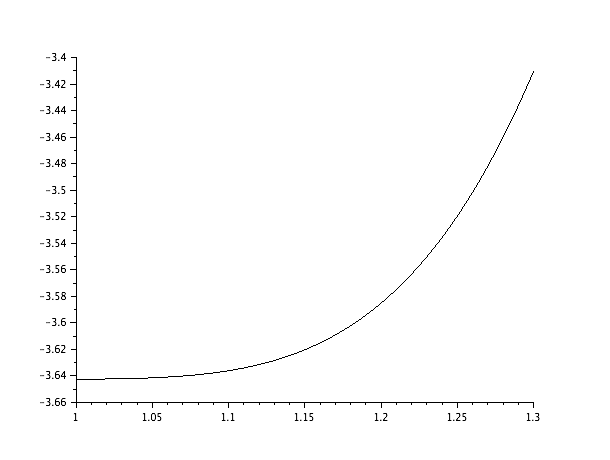} \includegraphics[width=8cm, trim=0cm 1mm 0cm 0cm,clip]{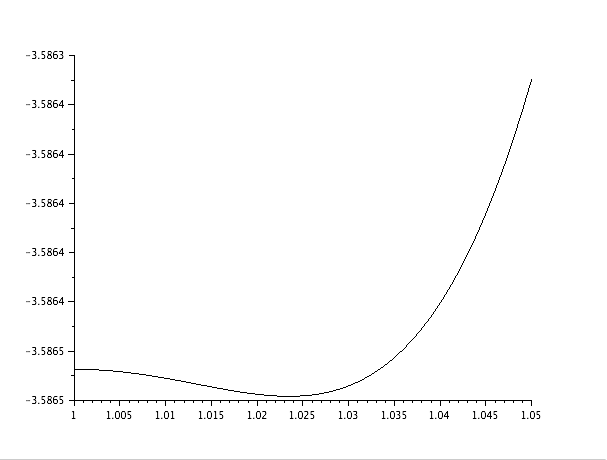}
\caption{Plots of $y\mapsto E_V(y,A)$ for $A=1.26$ (on the left) and $A=1.27$ (on the right). It seems that the minimizer is $y_A=1$ in the fist case and $y_A\approx 1.033$ in the second case.}
\label{Fig7}
\end{figure}

Numerically, in the domain Rect3, the minimizer seems to cover $(1,+\infty)$ monotonically and continuously with respect to $A$. In particular, the minimizer degenerates as $A$ goes to infinity, i.e. $\lim_{A\to +\infty} y_A=+\infty$ (see Figures \ref{Fig7}, \ref{Fig8} and \ref{Fig9}).\\

\begin{figure}[!h]
\centering
\includegraphics[width=8cm, trim=0cm 1mm 1mm 1mm, clip]{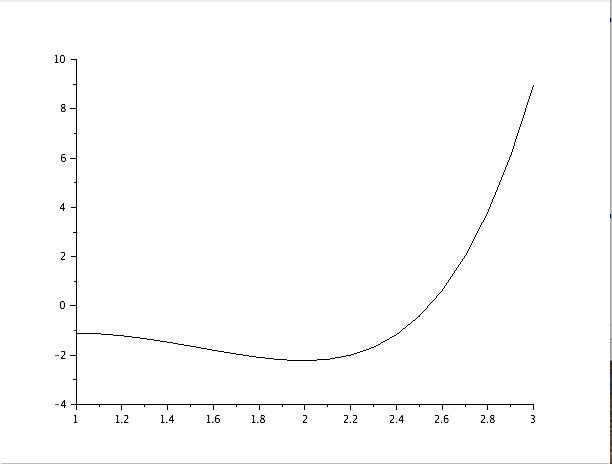} \includegraphics[width=8cm]{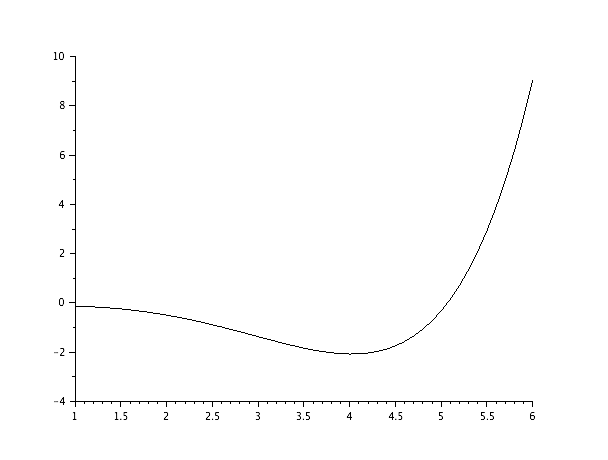}
\caption{Plots of $y\mapsto E_V(y,A)$ for $A=2$ (on the left) and $A=4$ (on the right).}
\label{Fig8}
\end{figure}

\begin{figure}[!h]
\centering
\includegraphics[width=8cm]{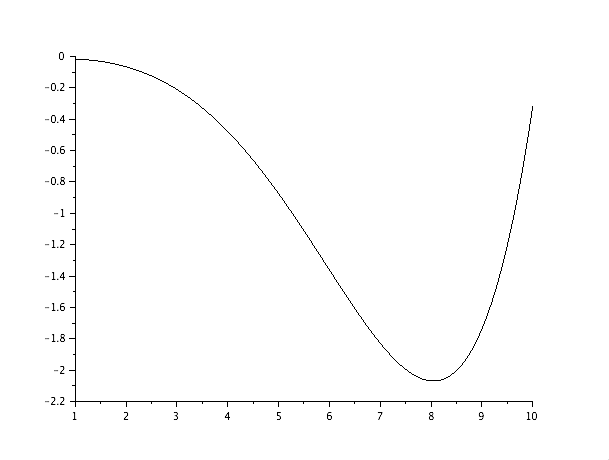} \includegraphics[width=8cm, trim=0cm 0cm 0cm 1mm,clip]{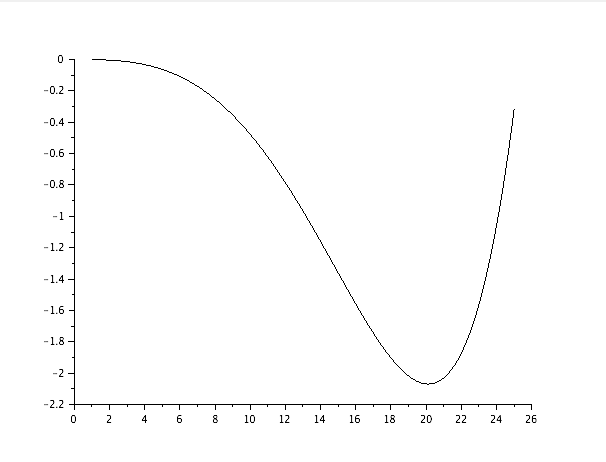}
\caption{Plots of $y\mapsto E_V(y,A)$ for $A=8$ (on the left) and $A=20$ (on the right).}
\label{Fig9}
\end{figure}

\begin{remark}
Heuristically, the degeneracy of the minimizer as $A\to +\infty$ follows from the fact that
$$
E_V[L_A]=\frac{1}{A^6}\left\{ \zeta_L(12)-2A^3\zeta_L(6) \right\}\sim -\frac{2}{A^3}\zeta_L(6)
$$
and the derivative, with respect to $x$, of the right-hand side expression is positive by Proposition \ref{genmgt}. Furthermore, the competition between $\zeta_L(12)$ and $-2A^3\zeta_L(6)$ is naturally won by the first one as $A\to 0$, and that explains why the triangular lattice is the minimizer for $A\in \textnormal{Rect1}$.
\end{remark}

The following results prove the degeneracy of the minimizer among rectangular lattices, as $A\to +\infty$, and the fact that the minimizer of $L_A\mapsto E_V[L_A]$ is rectangular if $A$ is large enough.

\begin{prop}[Degeneracy in the rectangular case]\label{degrect}
There exists $A_3>0$ such that, if $y_A\geq 1$ is a minimizer of $y\mapsto E_V(y;A)$, then for any $A>A_3$,
\begin{equation*}
X_1(A)^{1/3}\leq y_A\leq X_2(A)^{1/3}
\end{equation*}
where
$$
X_1(A)=\frac{2\zeta_{\Z^2}(6)A^3-\sqrt{4\zeta_{\Z^2}(6)^2A^6-32A^4+8\zeta_{\Z^2}(12)A^2}}{4}
$$
and
$$
X_2(A)=\frac{2\zeta_{\Z^2}(6)A^3+\sqrt{4\zeta_{\Z^2}(6)^2A^6-32A^4+8\zeta_{\Z^2}(12)A^2}}{4}
$$
In particular, 
\begin{enumerate}
\item we have $\displaystyle \lim_{A\to +\infty} y_A=+\infty$;
\item more precisely, there exists $C>0$ such that, for any $A>A_3$, $y_A\leq CA$.
\end{enumerate}
\end{prop}
\begin{proof} Let $A>0$ and $y\geq 1$, then we have
$$
 E_V(y,A)=A^{-6}\left(y^6\sum_{m,n}\frac{1}{(m^2+y^2n^2)^6}-2y^3A^3 \sum_{m,n}\frac{1}{(m^2+y^2n^2)^3} \right).
$$
Let $y_A$ be a minimizer, then we have $E_V(y_A,A)\leq E_V(A^{1/3},A)$, that is to say
\begin{align*}
y_A^6\sum_{m,n}\frac{1}{(m^2+y_A^2n^2)^6}-2y_A^3A^3 \sum_{m,n}\frac{1}{(m^2+y_A^2n^2)^3}\leq A^2\sum_{m,n}\frac{1}{(m^2+A^{2/3}n^2)^6}-2A^4 \sum_{m,n}\frac{1}{(m^2+A^{2/3}n^2)^3}.
\end{align*}
We remark that, for any $s\in\{3,6\}$ and $\alpha\geq 1$,
$$
2\leq \sum_{m,n} \frac{1}{(m^2+\alpha n^2)^s}\leq \zeta_{\Z^2}(2s).
$$
Thus, we get for $A\geq 1$,
$$
-4A^4+2y_A^3\zeta_{\Z^2}(6)A^3+\zeta_{\Z^2}(12)A^2-2y_A^6\geq 0.
$$
In particular, this inequality fails if $A$ is large enough. Indeed, we can rewrite this inequality as $R_A(y_A^3)\geq 0$ where $R_A$ is defined by
$$
R_A(X):=-2X^2+2\zeta_{\Z^2}(6)A^3 X+\zeta_{\Z^2}(12)A^2-4A^4.
$$
The discriminant of the polynomial $R_A$ is
$$
\Delta_A=4\zeta_{\Z^2}(6)^2A^6+8(\zeta_{\Z^2}(12)A^2-4A^4).
$$
Thus, if $A$ is sufficiently large, then $0<\Delta_A<4\zeta_{\Z^2}(6)^2A^6$. If follows that $R_A$ admits two positive zeros if $A$ is large enough, which are $1\leq X_1(A)<X_2(A)$ given in the statement of the proposition. Therefore, $R_A(y_A^3)\geq 0$ implies that, for $A$ sufficiently large, $X_1(A)\leq y_A^3\leq X_2(A)$ with
$$
X_1(A)^{1/3}\sim C_1 A^{1/3}\quad \textnormal{and}\quad X_2(A)^{1/3}\sim C_2 A,
$$
as $A\to +\infty$, where $C_1$ and $C_2$ are both positive constants, and the result is proved.
\end{proof}

\begin{remark}
It is crystal clear that the same result holds for all the Lennard-Jones type potentials.
\end{remark}

The following result shows why the minimizer is rectangular if $A$ is large enough.

\begin{prop}[The minimizer is rectangular at sufficiently low density]\label{Rankinmethod} 
There exists $A_4>0$ such that for any $A> A_4$, a minimizer $(x_A,y_A)$ of $(x,y)\mapsto E_V(x,y,A)$ satisfies $x_A=0$, i.e. any minimizer is a rectangular lattice.
\end{prop}
\begin{proof}
Let us prove that, for $A$ sufficiently large and any $(x,y)\in\mathcal{D}$, $\partial_x E_V(x,y,A)\geq 0$ with equality if and only if $x=0$. Using Rankin's notations \cite[Section 4., p. 157]{Rankin} and the notation $\zeta(x,y,s,A)=\zeta_{L_A}(s)$, we get, for any $x\neq 0$,
\begin{align*}
&A^6\partial_x E_V(x,y,A)\\
&=\partial_x\zeta(x,y,12,1)-2A^3\partial_x\zeta(x,y,6,1)\\
&=\frac{16\sqrt{\pi}}{4\sin^2\pi x}\sum_{k=1}^{+\infty}\left(\frac{C_1}{y^3} \Lambda(k,y,3)A^3-\frac{C_2}{y^6}\Lambda(k,y,6)  \right)\left\{ (k+1)\sin2\pi x -\sin2\pi(k+1)x \right\}\\
&=\frac{16\sqrt{\pi}}{4y^6\sin^2\pi x}\sum_{k=1}^{+\infty}\left(C_1y^3 \Lambda(k,y,3)A^3-C_2\Lambda(k,y,6)  \right)\left\{ (k+1)\sin2\pi x -\sin2\pi(k+1)x \right\}
\end{align*}
where
\begin{itemize}
\item $C_1$ and $C_2$ are both positive constants;
\item $\Lambda(k,y,s):=\lambda_{k+2}(y,s)-2\lambda_{k+1}(y,s)+\lambda_k(y,s)$;
\item $\lambda_k(y,s):=\sigma_{1-2s}(k)(2\pi ky)^{s+1/2}K_{s-1/2}(2\pi k y)$;
\item $\displaystyle \sigma_k(n)=\sum_{d|n}d^k$;
\item $\displaystyle K_\nu(u)=\int_0^{+\infty}e^{-u\cosh t}\cosh(\nu t)dt$.
\end{itemize}
Rankin \cite[Eq. (21)]{Rankin} proved that $\Lambda(k,y,3)>0$ for any $k\geq 1$ and any $y\geq \sqrt{3}/2$. Furthermore, by definition, we can write $\Lambda(k,y,3)=y^{7/2}\tilde{\Lambda}(k,y,3)$ and $\Lambda(k,y,6)=y^{13/2}\tilde{\Lambda}(k,y,6)$, where $\tilde{\Lambda}(k,y,3)$ and $\tilde{\Lambda}(k,y,6)$ are of the same order with respect to $y$. Therefore, we get, for any $(x,y)\in\mathcal{D}$,
\begin{align*}
&A^6\partial_x E_V(x,y,A)\\
&=\frac{16\sqrt{\pi}\sqrt{y}}{4\sin^2\pi x}\sum_{k=1}^{+\infty}\left(C_1\tilde{\Lambda}(k,y,3)A^3-C_2\tilde{\Lambda}(k,y,6)  \right)\left\{ (k+1)\sin2\pi x -\sin2\pi(k+1)x \right\}.
\end{align*}
Thus, since (see Rankin \cite[p. 158]{Rankin})
$$
1\leq  \sigma_{1-2s}(r)<\zeta(2s-1), \quad s\in\{3,6\}
$$
and (see \cite[p. 81]{Mont})
$$
(k+1)\sin2\pi x -\sin2\pi(k+1)x\geq 0, \quad k\geq 1, 0\leq x\leq 1/2,
$$
with equality for any $k\geq 1$ if and only if $x=0$, we obtain that $A^6\partial_x E_V(x,y,A)$ is positive, for any $(x,y)\in\mathcal{D}$, for $A$ sufficiently large. Consequently, there exists $A_4$ such that for any $A>A_4$, 
$$
\partial_x E_V(x,y,A)\geq 0,
$$
with equality if and only if $x=0$. It follows that the minimizer of $(x,y)\mapsto E_V(x,y,A)$ is such that $x_A=0$ for any $A>A_4$.
\end{proof}

A summary of both previous results is:

\begin{corollary}
For any $A>0$, we call $(x_A,y_A)\in\mathcal{D}$ a minimizer of $(x,y)\mapsto E_V(x,y,A)$. Then:
\begin{enumerate}
\item for $A$ sufficiently large, $x_A=0$;
\item it holds $\displaystyle \lim_{A\to +\infty} y_A=+\infty$.
\end{enumerate} 
\end{corollary}

\begin{remark} It numerically appears that the minimizer of $(x,y)\mapsto E_V(x,y,A)$ on $\mathcal{D}$ is a rectangular lattice for any $A>A_1$.
\end{remark}

\subsection{Remarks about the global minimality}\label{globminrmk}

Using our previous work \cite{Betermin:2014fy}, we can prove the following result explaining why the $A=1$ case is fundamental for finding the global minimizer of the Lennard-Jones energy, among Bravais lattices, without an area constraint.

\begin{prop}\label{globmin}
If $(1/2,\sqrt{3}/2)$ is the unique minimizer of $(x,y)\mapsto E_V(x,y,1)$, then the global minimizer of $(x,y,A)\mapsto E_V(x,y,A)$ is unique and triangular.
\end{prop}
\begin{proof}
By \cite[Proposition 3.5]{Betermin:2014fy}, we know that  $\Lambda_A$ is a minimizer of $L\mapsto E_V[L]$ among Bravais lattices of fixed area $A$ if and only if 
$$
A\leq \inf_{|L|=1\atop L\neq \Lambda_1} \left(\frac{\zeta_L(12)-\zeta_{\Lambda_1}(12)}{2(\zeta_L(6)-\zeta_{\Lambda_1}(6))}  \right)^{1/3}.
$$ 
Furthermore, we proved in \cite[Proposition 4.1]{Betermin:2014fy} that the area of a global minimizer is smaller than $1$. Thus, if the triangular lattice is the unique minimizer among Bravais lattices of fixed area $1$, then it is the case for every fixed area $A$ such that $0<A<1$. Consequently, the minimizer of the energy is unique and triangular, because the minimum among dilated triangular lattices with respect to the area is unique (see \cite[Proposition 4.3]{Betermin:2014fy}).
\end{proof}

We numerically check that $(1/2,\sqrt{3}/2)$ seems to be the minimizer of $(x,y)\mapsto E_V(x,y,1)$, but a rigorous proof have to be done. A strategy could be the following:
\begin{enumerate}
\item By Rankin's method (see proof of Proposition \ref{Rankinmethod}), we find $\partial_x E_V(x,y,1)\leq 0$ for any $(x,y)\in \mathcal{D}$, with equality if and only if $x=1/2$;
\item By the same arguments as in Proposition \ref{degrect}, it is possible to prove that the minimizer of $y\mapsto E_V(1/2,y,1)$ on $[\sqrt{3}/2,+\infty )$ admits an upper bound $y_1$;
\item By the algorithmic method based on \cite[Lem. 4.19]{BeterminPetrache}, the minimizer is $y=\sqrt{3}/2$ on $[\sqrt{3}/2,y_1]$.
\end{enumerate}

While the first point seems difficult to prove by using classical estimates, the proofs of both other points are clear.

\subsection{Summary of our results, numerical studies and conjectures}\label{summary}

In this part, we summarize the expected behavior of the minimizer $(x_A,y_A)$ of $(x,y)\mapsto E_V(x,y,A)$ based on our theoretical and numerical studies of the energy among rhombic and rectangular lattices. The summary is given in Figure \ref{CONJ}. In the following description, we detail the proved results and the conjectures based on numerical investigations. 

\begin{enumerate}
\item For $0<A<\frac{\pi}{(120)^{1/3}}\approx 0.637$, the minimizer is triangular. This is proved in \cite[Theorem 3.1]{Betermin:2014fy}.
\item For $\frac{\pi}{(120)^{1/3}}<A<A_{BZ}\approx 1.138$, the minimizer seems to be triangular. This is only a numerical result. In particular, if we know that $A_{BZ}>1$, then the global minimizer of $L\mapsto E_V[L]$, without a density constraint, is unique and triangular (see Proposition \ref{globmin}).
\item For $A_{BZ}<A<A_0\approx 1.152$, the triangular lattice is a local minimizer by Theorem \ref{THmain1}.
\item For $A_{BZ}<A<A_1\approx 1.143$, the minimizer seems, numerically, to be a rhombic lattice. More precisely it covers continuously and monotonically the interval of angles $[76.43^\circ,90^\circ)$.
\item For $A_1<A<A_2\approx 1.268$, the square lattice is a local minimizer, by Theorem \ref{THmain2}. Furthermore, it numerically seems that the square lattice is the unique minimizer of the energy.
\item For $A>A_2$, it numerically seems that the minimizer is a rectangular lattice. For $A$ large enough, we give a proof of this fact in Proposition \ref{Rankinmethod}.
\item As $A\to +\infty$, the minimizer becomes more and more thin and rectangular: it degenerates. This is proved in Proposition \ref{degrect}.
\end{enumerate}

\begin{figure}[!h]
\centering
\includegraphics[width=16cm]{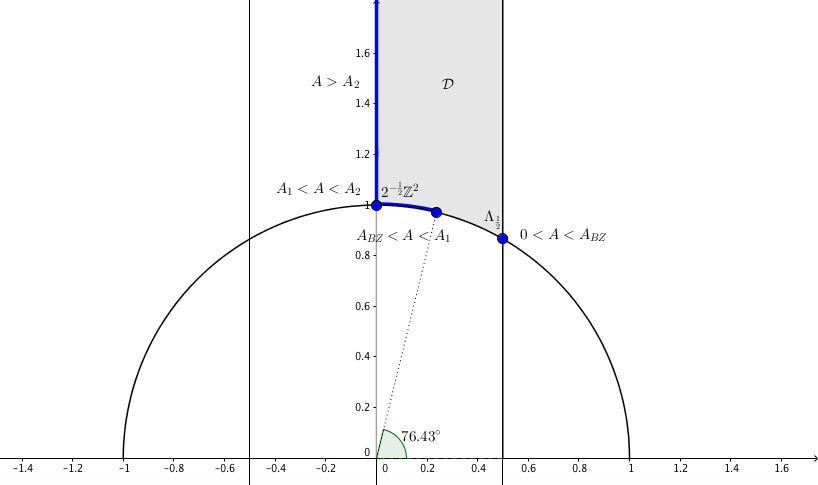} 
\caption{\textbf{Conjecture:} Behavior of the minimizer of $(x,y)\mapsto E_V(x,y,A)$ with respect to $A$.}
\label{CONJ}
\end{figure}

The evolution of the minimizer of $L_A\mapsto E_V[L_A]$ with respect to $A$ and the numerical investigations of Ho and Mueller \cite[Fig. 1 and 2]{Mueller:2002aa} (or see \cite{ReviewvorticesBEC}), for two-component Bose-Einstein Condensates are very similar. It is actually not surprising. Indeed, in their work, Ho and Muller consider the following lattice energy
$$
E_\delta(L,u):=\theta_L(1)+\delta\theta_{L+u}(1),
$$
among Bravais lattices $L\subset \R^2$ of area one and vectors $u\in \R^2$, where $\theta_L(\alpha)$ is defined by \eqref{defEpsttheta},  the translated theta function (see e.g. \cite[Sect. 1.B.]{BeterminPetrache}) is defined by
$$
\theta_{L+u}(1)=\sum_{p\in L} e^{-\pi |p+u|^2},
$$
and $-1\leq \delta \leq 1$. Thus, as we explained in \cite{BeterminPetrache}, this energy is the sum of two energies with opposite properties:
\begin{enumerate}
\item $L\mapsto \theta_{L}(1)$ is minimized by the triangular lattice $\Lambda_1$.
\item For any $u\not\in L$, $\theta_{L+u}(1)<\theta_{L}(1)$ and $L\mapsto \theta_{L+u}(1)$ does not admit any minimizer. More precisely, there exists a sequence of rectangular lattices $(L_k)_k$ which degenerates, as explained in Section \ref{rectangular}, such that $\lim_{k\to +\infty} \theta_{L_k+c_k}(1)=0$, where $c_k$ is the center of the primitive cell of $L_k$.
\end{enumerate}
Hence, since, for $L_A=\sqrt{A}L$ where $L$ has a unit area,
$$
A^6E_V[L_A]=\zeta_L(12)-A^3\zeta_L(6),
$$
$\theta_L(1)$ and $\theta_{L+u}(1)$ can be compared respectively to $\zeta_L(12)$ and $-\zeta_L(6)$. Furthermore, $\delta$ can be compared to $A^3$. Increasing $\delta$ (respectively $A$), Ho and Mueller find, as in this paper, that $L\mapsto \argmin_L\{ \min_{(L,u)}\left\{ E_\delta (L,u)\right\}\}$ (respectively $\argmin_L \{E_V[L]\}$) is triangular for small values of the parameter, becoming rhombic (with a discontinuous transition), square and finally rectangular.

\medskip

It is actually natural to conjecture that:
\begin{itemize}
\item the behavior of the minimizers of $L\mapsto E_f[L_A]$ with respect to the area $A$ is qualitatively the same for all the Lennard-Jones type potentials;
\item more generally, we can imagine that we should find the same result for any potential $f$ written as
$$
f=f_1-f_2,
$$
where $f_1$ and $f_2$ are both completely monotone and $f$ has a well, i.e. $f$ is decreasing on $(0,a)$ and increasing on $(a,+\infty)$. Indeed, for any $i\in\{1,2\}$, $L\mapsto E_{f_i}[L_A]$ has the same properties as $L\mapsto \theta_L(\alpha)$, for any $\alpha>0$ (see Proposition \ref{genmgt}).
\end{itemize}

\textbf{Acknowledgement:} I would like to thank the Mathematics Center Heidelberg (MATCH) for support, Doug Hardin for giving me the intuition of the degeneracy of the minimizer for the Lennard-Jones interaction, Mircea Petrache and Lukas Schimmer for giving me some feedbacks and Florian Nolte for interesting discussions. I also acknowledge support
from ERC advanced grant Mathematics of the Structure of Matter (project
No. 321029) and from VILLUM FONDEN via the QMATH Centre of Excellence (grant
No. 10059). I finally thank the anonymous referees for their useful suggestions and comments, that greatly helped to improve the readability of the paper.

\bibliographystyle{plain}
\bibliography{locmin}
\end{document}